\documentclass[twocolumn,prl,twocolumn,superscriptaddress]{revtex4}
\usepackage{color}
\usepackage{amsmath,amsthm,mathrsfs,dsfont}
\usepackage{graphicx}
\usepackage{tikz}
\usepackage{xr}
\usepackage{xc}
\usepackage{float}
\usepackage{bbm}
\usepackage{epsfig}
\usepackage{verbatim}
\usepackage{array}
\usepackage{setspace}
\usepackage{bbding}
\usepackage{amssymb}
\usepackage{pifont}
\usepackage{braket}
\usepackage{tikz}

\usetikzlibrary{quantikz}

\usepackage{newtxtext,newtxmath}

\externalcitedocument[suppl-]{../suppl/low-depth0312}

\usepackage[unicode=true,
 bookmarks=true,bookmarksnumbered=false,bookmarksopen=false,
 breaklinks=false,pdfborder={0 0 1},backref=false,colorlinks=true]
 {hyperref}
\hypersetup{
 linkcolor=magenta, urlcolor=blue, citecolor=blue, pdfstartview={FitH}, hyperfootnotes=false, unicode=true}

\usepackage{color}

\newtheorem{lemma}{Lemma}
\newtheorem{theorem}{Theorem}

\makeatletter
\@ifundefined{textcolor}{}
{%
 \definecolor{BLACK}{gray}{0}
 \definecolor{WHITE}{gray}{1}
 \definecolor{RED}{rgb}{1,0,0}
 \definecolor{GREEN}{rgb}{0,1,0}
 \definecolor{BLUE}{rgb}{0,0,1}
 \definecolor{CYAN}{cmyk}{1,0,0,0}
 \definecolor{MAGENTA}{cmyk}{0,1,0,0}
 \definecolor{YELLOW}{cmyk}{0,0,1,0}
}

\usepackage{amsfonts}\usepackage{tabularx}\usepackage{dcolumn}\usepackage{bm}\usepackage{graphicx}\usepackage{epstopdf}

\hypersetup{urlcolor=blue}

\makeatother

\newcolumntype{C}[1]{>{\centering\arraybackslash$}p{#1}<{$}}

\usepackage{algorithm}
\usepackage{algorithmic}

\begin{document}

\widetext
\title{Quantum State Preparation with Optimal Circuit Depth: Implementations and Applications}
\author{Xiao-Ming Zhang}
\affiliation{Center on Frontiers of Computing Studies, Peking University, Beijing 100871, China}
\affiliation{School of Computer Science, Peking University, Beijing 100871, China}

\author{Tongyang Li}
\affiliation{Center on Frontiers of Computing Studies, Peking University, Beijing 100871, China}
\affiliation{School of Computer Science, Peking University, Beijing 100871, China}

\author{Xiao Yuan}
\email{xiaoyuan@pku.edu.cn}
\affiliation{Center on Frontiers of Computing Studies, Peking University, Beijing 100871, China}
\affiliation{School of Computer Science, Peking University, Beijing 100871, China}

\begin{abstract}

Quantum state preparation is an important subroutine for quantum computing. We show that any $n$-qubit quantum state can be prepared with a $\Theta(n)$-depth circuit using only single- and two-qubit gates, although with a cost of an exponential amount of ancillary qubits. On the other hand, for sparse quantum states with $d\geqslant2$ nonzero entries, we can reduce the circuit depth to $\Theta(\log(nd))$ with $O(nd\log d)$ ancillary qubits. The algorithm for sparse states is exponentially faster than best-known results and the number of ancillary qubits is nearly optimal and only increases polynomially with the system size.  We discuss applications of the results in different quantum computing tasks, such as Hamiltonian simulation, solving linear systems of equations, and realizing quantum random access memories, and find cases with exponential reductions of the circuit depth for all these three tasks.
In particular, using our algorithm, we find a family of linear system solving problems enjoying exponential speedups, even compared to the best-known quantum and classical dequantization algorithms.

\end{abstract}
\maketitle

The speed limit of quantum state preparation  is a question with fundamental and practical interests, determining the efficiency of inputting classical data into a quantum computer, and playing as a critical subroutine for many quantum algorithms, such as in machine learning~\cite{Harrow.09,Lloyd.14,Kerenidis.16} and Hamiltonian simulations~\cite{Childs.18,Low.19}.
Without ancillary qubits, an exponential circuit depth is inevitable to prepare an arbitrary quantum state~\cite{Barenco.95,Knill.95,Vartiainen.04,Shende.04,Shende.06,Long.01,Grover.02,Mottonen.05,Bergholm.05,Plesch.11,Iten.16} and the optimal result $\Theta(2^n/n)$ was recently obtained by~\textcite{Sun.21_new}. Leveraging ancillary qubits, the circuit depth could be reduced to be subexponential scaling~\cite{Low.18,Zhang.21,Zhicheng.21,Sun.21_new,Rosenthal.21,Yuan.22,Clader.22}, yet in the worse case with an exponential number of ancillas. Very recently, the optimal circuit depth $\Theta(n)$ was achieved by~\cite{Sun.21_new, Rosenthal.21} with $O(2^n)$~\cite{Sun.21_new} and  $\tilde{O}(2^{n})$~\cite{Rosenthal.21} ancillary qubits.

Despite the previous results in minimizing circuit depth, the subexponential circuit depth is only achieved at the cost of exponential space complexity. 
Moreover, when considering applications in the field of quantum machine learning, strong data structure assumptions leave space for quantum-inspired classical algorithms. With a classical data structure enabling $l^2$ sampling, there are classical algorithms with polylogarithmic runtime dequantizing the quantum algorithms for recommendation systems~\cite{Tang.19}, solving linear systems~\cite{Chia.20,Gilyen.20}, semidefinite programs~\cite{Chia.20b}, etc.
These results show that space resources should not be neglected when discussing the quantum exponential advantages. 

In practice, the data may behave with a certain structure. Indeed, if the one imposes certain restrictions on the target quantum states, the circuit depth and the ancillary qubit number might be further reduced~\cite{Malvetti.21,Veras.21,Gleinig.21,Veras.20,Gabriel.21,Araujo.21,Rattew.22}. A typical scenario that has both theoretical and practical relevance is the sparse data structure, such as sparse classical data, Hamiltonians of physics systems, etc.
Using a constant number of ancillary qubits, arbitrary $d$-sparse quantum states (with $d$ nonzero entries)
can be prepared using a circuit depth of $O(dn)$
~\cite{Malvetti.21,Veras.21,Gleinig.21,Veras.20}. However, it was unclear if the sparse preparation procedure could be further sped up with more, but polynomial, ancillary qubits.   The fundamental speed limit of sparse state preparation is still an open question, which is important for studying the ultimate power of quantum machine learning algorithms.

In this work, we study the speed limit of quantum state preparation. 
We first develop a deterministic algorithm (independent of Refs.~\cite{Sun.21_new,Rosenthal.21}) for preparing an arbitrary quantum state with optimal circuit depth $\Theta(n)$ and $O(2^n)$ ancillary qubits.
The scheme requires a much more sparse connectivity than Ref.~\cite{Sun.21_new}, as each qubit connects to a constant number of other qubits.
We next introduce an algorithm for $d$-sparse quantum states ($d\geqslant 2$) that  achieves the optimal circuit depth $\Theta(\log(nd))$, exponentially faster than the best-known results~\cite{Malvetti.21,Veras.21,Gleinig.21}. The sparse state preparation requires $O(nd\log d)$ ancillary qubits, which is also nearly optimal.
Based on the results, we find a family of linear system tasks that can be solved with the circuit depth and the number of ancillary qubits being $O(\text{poly}(n))$, and hence show an exponential improvement compared to the best known quantum and classical dequantization algorithms.
We also show how our techniques can be applied to improving Hamiltonian simulations and quantum random access memories (QRAMs).

\begin{figure}
	\centering
	\includegraphics[width=1\columnwidth]{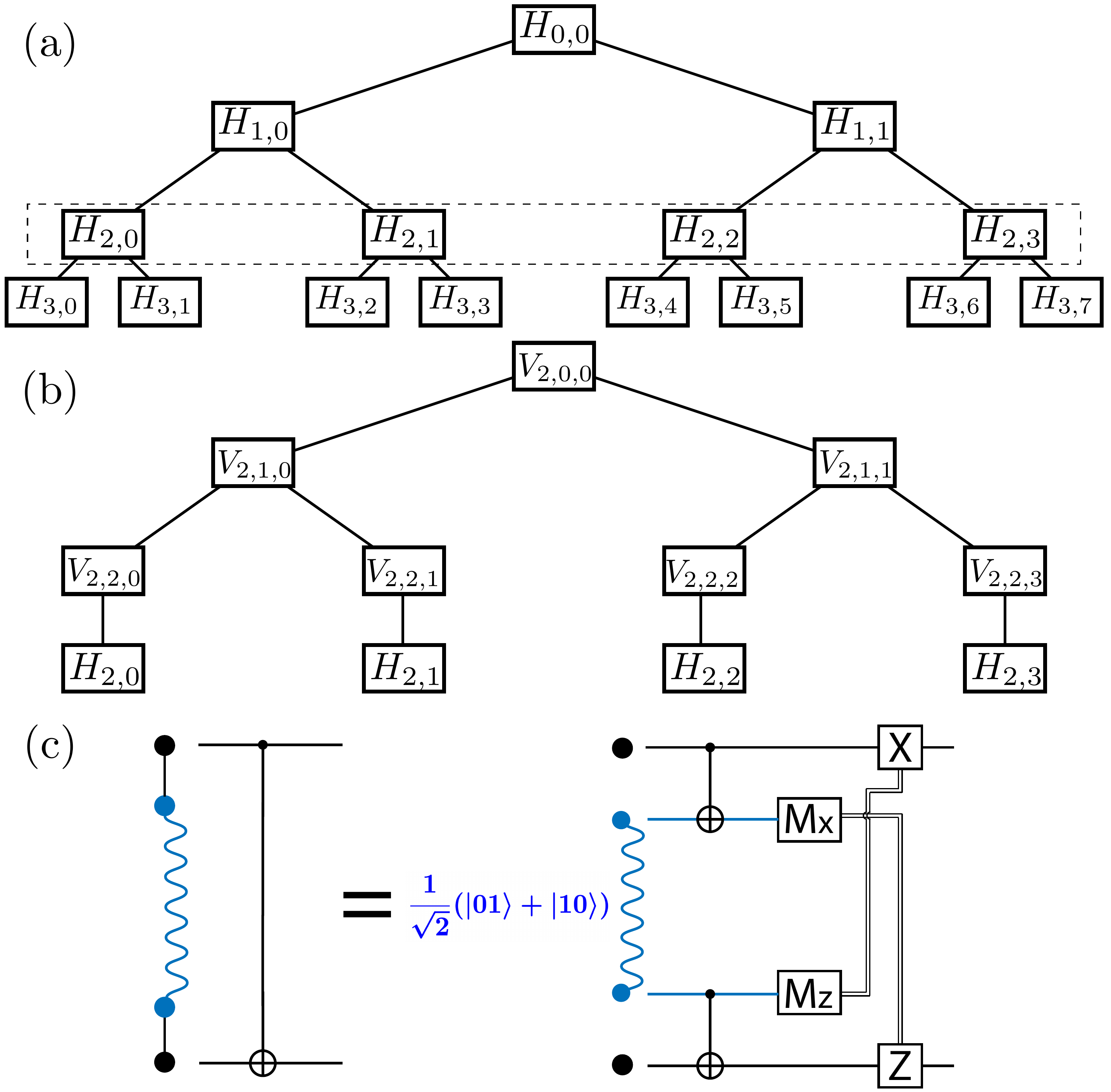}
	\caption{(a) Layout of binary tree $H$. Each block represents a qubit. (b) Layout of binary tree $V_2$, which connects to the second layer of $H$ with dashed box, i.e. $H_2$ . Here, $V_{2,\text{root}}$ is $V_{2,2,0}$. In (a) and (b), CNOT gates are only applied at qubit pairs connected by solid lines. (c) CNOT gate between two distant qubits (black circles) based on preshared Bell states (blue circles). $M_{x,z}$ and  $X$, $Z$  represent measurements and Pauli gates~\cite{Chou.18}. 
	} \label{fig:ab}
\end{figure}

\textbf{\textit{Access model.---}}A general $n$-qubit state can be  expressed as
\begin{align}\label{eq:ab}
|\psi\rangle=\sum_{k=0}^{N-1}a_k|k\rangle,
\end{align}
with $N=2^n$, $a_k\in\mathbb{C}$,  $\sum_{k=0}^{N-1} |a_k|^2=1$, and $|k\rangle\equiv|k_nk_{n-1}\cdots k_1\rangle$ being the  basis with bits $k_j$ for $j=1,2,\dots,n$. Before discussing our state preparation protocol, we first introduce how our quantum circuit accesses the classical description of a target quantum state.
Let $b_{n,k}\equiv|a_{k}|$, $b_{l,j}\equiv\sqrt{|b_{l+1,2j}|^2+|b_{l+1,2j+1}|^2}$ for $0\leqslant l\leqslant n-1$, $\theta_{l,j}= \arccos (b_{l,2j}/b_{l-1,j})$ for $b_{l-1,j}\neq0$, and $\theta_{l,j}=0$ for $b_{l-1,j}=0$. We require classical preprocessing to calculate $\theta_{l,j}$, and $\text{arg}(a_k)$. Here, $b_{l,j}$ are recursively defined so that we can encode the amplitudes in a treelike fashion allowing parallelization. This recursive definition is not required for phase arg$(a_k)$, because after encoding the amplitude, the phase can be encoded with a single layer of phase gates (see Sec.~I of~\cite{sm} for details). 
The preprocessing takes time $O(N)$ by sequential calculations, or $O(\log N)$ by parallel calculations with $O(N)$ space complexity.
These complexities are optimal because reading and writing $N$ values already require $\Omega(N)$ resource. \\
\indent For sparse quantum state with $d$ nonzero elements, the quantum state can be expressed as
\begin{align}\label{eq:sp}
|\psi\rangle=\sum_{k=0}^{d-1}\psi_k|q_k\rangle,
\end{align}
where $\psi_k\in \mathbb{C}$ and $q_k$ is the index (with $n$ digits) of the $k$th nonzero entries.
We assume $d=2^{\tilde n}$ with integer $\tilde n$, which can be always satisfied by appending $|q_k\rangle$ with zero amplitude.
Similarly, we let $b'_{n,k}\equiv|\psi_{k}|$, $b'_{l,j}\equiv\sqrt{|b'_{l+1,2j}|^2+|b'_{l+1,2j+1}|^2}$ for $0\leqslant l\leqslant \tilde n-1$, and $\theta'_{l,j}\equiv \arccos (b'_{l,2j}/b'_{l-1,j})$, and require classical preprocessing to calculate $\theta'_{l,j}$, $\text{arg}(\psi_k)$. The value of $q_{k}$ should also be encoded to the circuit. The preprocessing time is $O(nd)$ for sequential calculation, or $O(\log(nd))$ for parallel calculation with $O(nd)$ space complexity.
\\
\indent The calculated angles and the labels of nonzero basis $|q_k\rangle$ for sparse states can then be directly mapped to the parameters of the quantum circuit, so the time complexity for generating quantum circuits are identical to the preprocessing time. Note that the preprocessing only needs to be performed once for preparing arbitrary copies of state. Here and after, we assume that the classical preprocessing has been completed.

\textbf{\textit{Quantum state preparation.---}}
Without loss of generality, the task of quantum state preparation is to prepare $|\psi\rangle$ from an initial product state $|0\rangle^{\otimes n}$ using single- and two-qubit gates. The qubit layout of our protocol is illustrated in Figs.~\ref{fig:ab}(a) and~\ref{fig:ab}(b). There is an ($n+1$)-layer binary tree of qubits, $H$. The $l$th layer of $H$ is denoted as $H_{l}$, and its $j$th qubit is denoted as $H_{l,j}$. The $l$th layer of $H$ is connected to the leaf layer of another binary tree $V_l$ with $(l+1)$ layers. In this layout, each qubit connects to at most $3$ of the other qubits, while Ref.~\cite{Sun.21_new} assumes that two-qubit gates can be applied on any two qubits. With the qubit layout above and the access model introduced previously, we have the following result.

\begin{theorem}[Arbitrary quantum state preparation]\label{th:ab}
With only single- and two-qubit gates, an arbitrary $n$-qubit quantum state can be deterministically prepared with a circuit depth $\Theta(n)$ and $O(N)$ ancillary qubits.
\end{theorem}
\noindent Our method saturates the  circuit depth lower bound $\Omega(n)$ ~\cite{Zhang.21,Sun.21_new}.
Below we sketch our protocol and refer to Sec.~I of~\cite{sm} for the formal description.

The root of $H$ is initialized as $|1\rangle$ and all other qubits are initialized as $|0\rangle$. The protocol contains $5$ stages. In stage $1$, with a $O(n)$ layer of CNOT and single qubit gates, $H$ is prepared as 
\begin{align}\label{eq:ab_2}
|\psi_{\text{stage 1}}\rangle=\sum_{k=0}^{2^{n}-1}a_{k}|1\rangle_{H_0}\bigotimes_{l=1}^{n}|(k,l)\rangle'_{H_l},
\end{align}
where $|1\rangle_{H_0}$ the state of $H_0$, and $|(k,l)\rangle'_{H_l}$ is the state of $H_l$. Here $(k,l)\equiv k_nk_{n-1}\cdots k_{n-l+1}$ represents the last $l$ digits of $k$, and $|(k,l)\rangle'\equiv|0\rangle^{\otimes (k,l)}|1\rangle|0\rangle^{\otimes 2^l-(k,l)-1}$. At each layer, there is only one qubit activated (at state $|1\rangle$) while the rest of the qubits are at state $|0\rangle$. The amplitude of the basis $|1\rangle_{H_0}\bigotimes_{l=1}^{n}|(k,l)\rangle'_{H_l}$ at Eq.~\eqref{eq:ab_2} is identical to the amplitude of $|k\rangle$ at Eq.~\eqref{eq:ab}. So the remaining task is all about basis transformation.

In stage $2$, for each $l$,  we map the state of $H_l$ to the state of the root of $V_{l}$. More specifically, we perform $|(k,l)\rangle'_{H_l}|0\rangle_{V_{l,\text{root}}}\rightarrow|(k,l)\rangle'_{H_l}|k_{n-l+1}\rangle_{V_{l,\text{root}}}$ where $V_{l,\text{root}}=V_{l,0,0}$. With a total circuit depth $O(n)$, we obtain
\begin{align}\label{eq:tree}
|\psi_{\text{stage 2}}\rangle=\sum_{k=0}^{2^{n}-1}a_{k}|1\rangle_{H_0}\bigotimes_{l=1}^{n}|(k,l)\rangle'_{H_l}|k_{n-l+1}\rangle_{V_{l,\text{root}}}.
\end{align}
In the remaining of the algorithm (stage $3$ to $5$), our goal is to uncompute $H$. This can be realized by flipping each qubit of $H_l$ conditioned on the states of its parent and $V_{l,\text{root}}$. By utilizing the binary trees $V_l$, the uncomputation can be done with $O(n)$ circuit depth~\cite{sm}.
We can then trace out all qubits except for the roots of $V_{l}$, and the state becomes $\sum_{k=0}^{2^n-1}\alpha_k\bigotimes_{l=1}^{n}|k_{n-l+1}\rangle_{V_{l,\text{root}}}$, which is equivalent to Eq.~\eqref{eq:ab}.

We also show in~\cite{sm} that our scheme can be approximated to accuracy $\varepsilon$ using Clifford+$T$ gates with depth $O(n\log(n/\varepsilon))$. This decomposition is important for fault-tolerant implementation based on surface code~\cite{Fowler.12}. 

\textbf{\textit{Nonlocal entangling gate implementation.---}} Similar to QRAMs, encoding exponential data to a quantum state may require spatially nonlocal gates. The nonlocal gate can be realized by quantum network with preshared Bell states~\cite{Gottesman.99,Chou.18}. As shown in Fig.~\ref{fig:ab}(c), the scheme requires a pair of ancillary qubits at Bell state $1/\sqrt{2}(|01\rangle+|10\rangle)$. Each ancillary qubit couples to either the control or target qubit. Effective CNOT can then be realized with local operations and classical communication. The protocol has been demonstrated in superconducting qubit~\cite{Chou.18} and trapped-ion systems~\cite{Wan.19}. Alternatively, nonlocal gates can also be realized with spin-photon network~\cite{Chen.21}. In fault-tolerant settings, surface code based on teleportated CNOT gate above has also been proposed in~\cite{Xu.22}, which can be straightforwardly applied to our scheme. In~\cite{sm}, we further show that our scheme can be implemented even in a nearest-neighbor coupled two-dimensional qubit array, only at the cost of a mild increase of the ancillary qubit number to $O(n^2N)$.

\textbf{\textit{Sparse quantum state preparation.---}}
The protocol can be further improved if the target states are sparse. We first introduce two subroutines that are useful for sparse state preparation and then discuss several other applications. 
Both subroutines work on a quantum system containing an \textit{index register} and a \textit{word register}, which are systems with certain number of qubits. 

The first subroutine is the product unitary memory (PUM), which can be considered as a generalization of QRAM protocol in~\cite{Hann.19,Hann.21}. We define an $n$-word product unitary function $\hat U(k)\equiv \bigotimes_{l=n}^{1} \hat U_l(k)$ with $\hat U_l(k)\in \text{SU}(2)$ and $k\in\{0,1\cdots,d-1\}$.
We define the selector unitary of $\hat U$ as select$(\hat U)$, which satisfies select$(\hat U)|k\rangle|z\rangle=|k\rangle \hat U(k)|z\rangle$. Here, $|k\rangle$ is the basis of the $\lceil\log_2 d\rceil$-qubits index register, and $|z\rangle$ is the basis of the $n$-qubits word register. The selector unitary can also be represented as $\text{select}(\hat U)\equiv\sum_{k=0}^{d-1}|k\rangle\langle k|\otimes \hat U(k)$. We have the following result (see Sec. III A of~\cite{sm} for details).

\begin{lemma}[PUM]\label{lm:1}
Given an arbitrary $\lceil\log_2d\rceil$-index, $n$-word product unitary function $\hat U(k)$, select$(\hat U)$ can be realized with circuit depth $O(\log (nd))$ and $O(nd)$ ancillary qubits using only single- and two-qubit gates.
\end{lemma}

The second subroutine is the sparse Boolean memory (SBM).
We consider an $n$-index, $\tilde{n}$-word Boolean function $f:\{0,1\}^n\rightarrow \{0,1\}^{\tilde{n}}$. Let $\mathcal{S}_{f}=\{k|f(k)\neq 0\cdots0\}$ containing all input indexes with nonzero output. We say that $f$ is $s$ sparse if $\mathcal{S}_f$ has no more than $s$ elements. Its corresponding sparse Boolean function selector satisfies select$(f)|k\rangle|z\rangle=|k\rangle|z\oplus f(k)\rangle$, where $\oplus$ represents bitwise XOR. Let $f_l(k)$ be the $l$th digit of $f(k)$, select$(f)$ can also be expressed as
\begin{align}\label{eq:selef}
\text{select}(f)\equiv\sum_{k=0}^{2^n-1}|k\rangle\langle k|\bigotimes_{l=\tilde{n}}^{1} \left(f_l(k)\hat X+\overline f_l(k)\hat{\mathbb{I}}_1\right),
\end{align}
 where $\overline f_l(k)$ is the NOT of $f_l(k)$, $\hat X$ is the Pauli-$X$ operator and $\hat{\mathbb{I}}_{m}$ represents the $m$-qubit identity. We have the following result (see Sec. III A of~\cite{sm} for details).

\begin{lemma}[SBM]\label{lm:2}
Given an arbitrary $n$-index, $\tilde n$ word, $s$-sparse Boolean function $f$, select$( f)$ in Eq.~\eqref{eq:selef} can be {realized with a quantum circuit with circuit depth} $O(\log (ns\tilde n))$ and $O(ns\tilde n)$ ancillary qubits using only single- and two-qubit gates.
\end{lemma}

Based on Lemma~\ref{lm:1},~\ref{lm:2}, and the access model discussed previously, we are now ready for our sparse state preparation protocol. Our result is as follows.

\begin{theorem}[Sparse state preparation]\label{th:sp}
With only single- and two-qubit gates, arbitrary $n$-qubit, $d$-sparse ($d\geqslant2$) quantum states can be deterministically prepared with a circuit depth $\Theta(\log(nd))$ and $O(nd\log d)$ ancillary qubits.
\end{theorem}

As we prove in Lemma~3 at Sec.~II of \cite{sm} that the circuit depth is lower bounded by $\Omega(\log(nd))$, our protocol also achieves the optimal circuit depth for sparse states.         
 
Below, we show how our protocol works for preparing $\sum_{k=0}^{d-1}\psi_k|q_k\rangle$ in Eq.~\eqref{eq:sp}.
We introduce registers $A$ and $B$, consisting of $\tilde n=\lceil \log_2d \rceil$ and $n$ qubits respectively. All qubits are initialized to $|0\rangle$. Then, we  prepare register $A$ to state $\sum_{k=0}^{d-1}\psi_k|k\rangle_A$, which uses $O(\log d)$ circuit depth and $O(d)$ ancillary qubits according to Theorem~\ref{th:ab}.
Next, we introduce an $n$-word product unitary function $\hat U_{\text{prep}}(k)=\bigotimes_{j=n}^{1}\left(q_{k,j}\hat{X}+\overline q_{k,j}\hat{\mathbb{I}}_1\right)$, where $q_{k,j}$ is the $j$th digit of $q_k$. We query select $\left(\hat U_{\text{prep}}\right)$ with register $A$ as index register and register $B$ as word register, and the output state is
$
\sum_{k=0}^{d-1}\psi_k|k\rangle_A|q_k\rangle_B$.
According to Lemma~\ref{lm:1}, this step can be realized with $O(\log(nd))$ circuit depth and $O(nd)$ ancillary qubits.
The remaining procedure is to uncompute register $A$. To do so, we introduce another $n$-index, $\tilde n$-word Boolean function $f_{\text{prep}}$. Let $\mathcal{Q}$ be a set containing all nonzero entries $q_k$ of the target state. The definition of $f_{\text{prep}}$ is that $f_{\text{prep}}(q_k)=k$ for $q_k\in\mathcal{Q}$ and $f_{\text{prep}}(q)=0$ for $q\notin\mathcal{Q}$ (i.e. $\mathcal{S}_{f_{\text{prep}}}=\mathcal{Q}$). We query select$\left(f_{\text{prep}}\right)$ with register $B$ as index register and register $A$ as word register, after which the state becomes
$\sum_{k=0}^{d-1}\psi_k|0\cdots0\rangle_A|q_k\rangle_B$.
The target state is obtained by tracing out register $A$. 

Because $f_{\text{prep}}$ is $d$ sparse, according to Lemma~\ref{lm:2}, this step has circuit depth $O(\log (nd))$ and space complexity $O(nd\log d)$. So the total circuit depth and space complexity of sparse state preparation is $\Theta(\log (nd))$ and $O(nd\log d)$. Moreover, as mentioned previously, it takes classical runtime $O(nd)$ to generate the quantum circuit, which can be reduced to $O(\log(nd))$ for parallel calculation with $O(nd)$ space complexity.

Theorem~\ref{th:sp} also provides a method for approximating nonsparse states with sparse states. We denote $a^{\text{max}}_j$ as the $j$th largest value of $|a_k|$. Suppose $\sum_{j=1}^{d}|a_{j}^{\max}|^2=1-\varepsilon$, we can then set all amplitudes $|a_k|<a^{\text{max}}_d$ to zero and normalize the sparse state. According to Theorem~\ref{th:sp}, the quantum state can be approximated to fidelity $F=1-\varepsilon$ with circuit depth  $O(\log(nd))$ (see Sec.~IX of~\cite{sm} for details).

We next discuss applications of our results.

\textbf{\textit{Hamiltonian simulation.---}}
A Hamiltonian $\hat{H}$ can generally be expressed as a linear combination of products of single qubit unitaries (such as Pauli strings)
\begin{align}\label{eq:Hpauli}
\hat{H}=\sum_{p=0}^{P-1}\alpha_p\hat{V}(p),
\end{align}
for some $\alpha_p> 0$, $\hat V(p)=\bigotimes_{l=1}^{n}\hat V_l(p)$ and $\hat V_l(p)\in\text{SU}(2)$. 
Simulation of $e^{-iHt}$ with optimal complexity with respect to the accuracy can be achieved with block-encodings~\cite{Chakraborty.19,Low.19}.
We say that $\hat U$ is a block-encoding~\cite{Chakraborty.19} of $\hat{H}$ if $(\langle 0|^{\otimes a}\otimes \mathbb{I}_n)\hat{U}(|0\rangle^{\otimes a}\otimes \mathbb{I}_n)\propto \hat{H}$ for some integer $a$. One common construction way of block-encoding is based on linear combination of unitaries~\cite{Low.19}.
We define $\hat G$ as a quantum state preparation operator satisfying $\hat G|0\rangle=|G\rangle\equiv\sum_{p}\sqrt{\alpha_p/\alpha}|p\rangle$ with $\alpha\equiv\sum_{p}\alpha_p$. It can be verified that $(\hat G^\dag\otimes \mathbb{I}_n) \text{select}(\hat V)(\hat G\otimes \mathbb{I}_n)$ is a block-encoding of $\hat H$.
Conventional ways to implement select($\hat V$) and block-encoding requires a circuit depth $O(nP)$~\cite{Childs.18}.

In contrast, according to Theorem~\ref{th:ab} and Lemma~\ref{lm:1},  $\hat G, \hat G^\dag$ can be realized with circuit depth $O(\log P)$, and select$(\hat V)$ can be realized with circuit depth $O(\log(Pn))$. So the block-encoding can be constructed with circuit depth $O(\log(nP))$.
Combining qubitization~\cite{Low.19} with our fast construction of block-encoding, we have the following result (see Sec.~IV of~\cite{sm} for details), which reduce the circuit depth exponentially to respect to $nP$.

\begin{theorem}[Hamiltonian simulation by qubitization]\label{th:hs}
Let $\hat{H}$ be a Hermitian operator expressed as Eq.~\eqref{eq:Hpauli}. Using only single- and two-qubit gates, the evolution $e^{-i\hat H t}$ can be simulated to precision $\varepsilon$ with circuit depth $O\left(\log (nP)(\alpha t+\log(1/\varepsilon))\right)$ and $O(nP)$ qubits.
 \end{theorem}

We note that another version of parallel Hamiltonian simulation has been proposed in Ref.~\cite{Zhicheng.21}, achieving doubly logarithmic circuit depth with respect to the precision $O(\log^3\log(1/\varepsilon))$. The algorithm is based on a state preparation method with
cubic circuit depth. In Sec.~V of~\cite{sm}, we show that the circuit depth can be further reduced to $O(\log^2\log(1/\varepsilon))$ based on Theorem.~\ref{th:ab}.

\textbf{\textit{Solving linear systems.---}}
Given an invertible matrix $H\in R^{2^n\times 2^n}$ and vector $b\in R^{2^n}$, quantum algorithms of  linear systems aim at generating an approximation of quantum state $|x\rangle$ proportional to $H^{-1}\cdot b$. For sparse $H$, $|x\rangle$ can be obtained with a circuit depth $O(\text{poly}(n))$~\cite{Harrow.09,Childs.17,Chakraborty.19}. The results has also been generalized to nonsparse cases~\cite{Wossnig.18}.
However, these quantum algorithms assume the query of quantum state preparation and Hamiltonian simulation oracles. In general, to achieve poly-logarithmic circuit depth, data structure with space complexity $O(2^n)$ is required, leaving room for classical dequantization algorithms~\cite{Tang.19,Chia.20,Gilyen.20}. Specifically, based on an analog data structure with $O(\text{nnz}(H)\cdot n)$ space complexity ($\text{nnz}(\cdot)$ refers to the number of nonzero entries), classical  dequantization algorithms~\cite{Chia.20,Gilyen.20} can sample from the distribution of the measurement outcomes of $|x\rangle$ with a circuit depth $O(\text{poly}(n))$.

Therefore, whether we could more efficiently solve linear systems heavily relies on efficiency of the quantum state preparation and Hamiltonian simulation oracles. Here, considering sparse matrices of Eq.~\eqref{eq:Hpauli}, we show an exponential advantage of quantum computing algorithms based on Theorem.~\ref{th:sp},~\ref{th:hs}.

\begin{theorem}[Solving linear system]\label{th:ls}
Let $\hat{H}$ be a Hermitian expressed as Eq.~\eqref{eq:Hpauli} with  condition number $\kappa$. Let $|b\rangle$ be a $O(\text{poly}(n))$-sparse quantum state. With only single- and two-qubit gates, the quantum state $|x\rangle$ proportional to $H^{-1}|b\rangle$ can be approximately prepared to precision $\varepsilon$ using {$\tilde O(\text{poly}(\log(nP),\alpha,\kappa))$}
circuit depth and {$ O(\text{poly}(n,P))$} qubits, where $\tilde O$ neglects the logarithmic dependence on $\kappa,1/\varepsilon$.
\end{theorem}

With $P,\alpha=O(\text{poly}(n))$, our method has $\tilde O(\text{poly}(n,\kappa))$ circuit depth and $O(\text{poly}(n))$ space complexity.
In comparison, the data structures by classical dequantization algorithms~\cite{Chia.20,Gilyen.20} have $ O(\text{nnz}(\hat H)\log N)=O(NPn)$ space complexities, which is exponentially larger.
Furthermore, assuming that $P,\alpha=O(1)$ and $|b\rangle$ is $O(1)$-sparse,  the circuit depth of our parallel method is further reduced to {$\tilde O(\log(n)\text{poly}(\kappa))$}, which is also exponentially lower than the depth {$\tilde O(n\,\text{poly}(\kappa))$} of the quantum algorithms using sequential select$(\hat H)$~\cite{Berry.15,Childs.18}; Yet, both of them have $O(n)$ space complexity. More details are provided in Sec.~VI of~\cite{sm}.

\textbf{\textit{QRAMs.---}}
At last, we show applications for QRAMs.
Given a binary dataset $\mathcal{D}\equiv \left\{D_k\right\}_{k=0}^{2^n-1}\in\{0,1\}^n$, QRAMs are memory architectures enabling the transformation
\begin{align}\label{eq:qram_0}
\text{QRAM}(\mathcal{D})\sum_{k=0}^{2^n-1}\psi_k|k\rangle|0\rangle=\sum_{k=0}^{2^n-1}\psi_k|k\rangle|D_k\rangle.
\end{align}
Efficient implementation of QRAM is important for many applications, especially for quantum machine learning~\cite{Biamonte.17}.
Conventional methods have $O(n)$ circuit depth using $O(N)$ ancillary qubits~\cite{Giovannetti.08,Giovannetti.08_2,Hann.19,Hann.21}. If $\mathcal{D}$ is $d$ sparse (with at most $d$ nonzero elements), the space complexity can be significantly reduced using quantum read-only access memory (QROM)~\cite{Babbush.18}. Alternatively, Eq.~\eqref{eq:qram_0} can be realized by performing $n$-Toffoli gates for $d$ times. However, these methods have circuit depth linear with $d$, which is not yet optimal.

On the other hand, by defining  $f_{\text{qram}}(k)\equiv D_k$, Eq.~\eqref{eq:qram_0} can be satisfied by select$(f_{\text{qram}})$. According to Lemma~\ref{lm:2},  we can obtain the following result.
\begin{theorem}[Sparse QRAM]\label{th:sqram}
With only single- and two-qubit gates, arbitrary $\text{QRAM}(\mathcal{D})$ in Eq.~\eqref{eq:qram_0} with $d$-sparse $\mathcal{D}$ can be implemented with circuit depth $O(\log(nd))$ and $O(nd)$ ancillary qubits.
\end{theorem}
\noindent Our protocol thus has an exponentially lower circuit depth compared to existing ones~\cite{Giovannetti.08,Giovannetti.08_2,Babbush.18,Hann.19,Hann.21}, while the space complexity remains polynomial.

Moreover, we can construct a nonsparse QRAM for continuous amplitude, i.e. $|D_k\rangle\in\mathbb{C}^2$, based on Lemma~\ref{lm:1} (see Sec.VII of~\cite{sm}). Our method requires $O(n)$ circuit depth and $O(N)$ ancillary qubits. The connectivity is identical to those for binary QRAMs~\cite{Giovannetti.08,Giovannetti.08_2,Hann.19,Hann.21}, which is more sparse than other continuous QRAM schemes developed recently~\cite{Yuan.22,Clader.22}  assuming all-to-all connectivity.

\textbf{\textit{Discussions.---}} 
We have achieved optimal circuit depth for general and sparse quantum state preparation. 
While Theorem~\ref{th:ab},~\ref{th:sp} assume that only single- and two-qubit gates are allowed, our results can be generalized to constant weight operations, i.e. operations applied to constant number of qubits. It therefore represents a fundamental limit of quantum information processing imposed by constant-weight operations. Future direction includes finding optimal space-time trade-offs for sparse state preparation, and exploring more practical applications.

\begin{acknowledgments}

\noindent \textbf{\emph{Acknowledgments.---}}
We thank Bujiao Wu for insightful discussions. This work is supported by the National Natural Science Foundation of China Grant No.~12175003, and Emerging Engineering Interdisciplinary-Young Scholars Project, Peking University, the Fundamental Research Funds for the Central Universities.
\end{acknowledgments}

%

\vspace{1cm}
\onecolumngrid
\newpage

\begin{center}
{\bf\large Supplementary material}
\end{center}
\vspace{0.5cm}

\setcounter{secnumdepth}{3}  
\setcounter{equation}{0}
\setcounter{figure}{0}
\setcounter{table}{0}
\setcounter{section}{0}

\renewcommand{\theequation}{S-\arabic{equation}}
\renewcommand{\thefigure}{S\arabic{figure}}
\renewcommand{\thetable}{S-\Roman{table}}
\renewcommand\figurename{Supplementary Figure}
\renewcommand\tablename{Supplementary Table}
\newcommand\citetwo[2]{[S\citealp{#1}, S\citealp{#2}]}
\newcommand\citecite[2]{[\citealp{#1}, S\citealp{#2}]}

\newcolumntype{M}[1]{>{\centering\arraybackslash}m{#1}}
\newcolumntype{N}{@{}m{0pt}@{}}

\makeatletter \renewcommand\@biblabel[1]{[S#1]} \makeatother

\makeatletter \renewcommand\@biblabel[1]{[S#1]} \makeatother



\section{Arbitrary quantum state preparation}\label{app:ab}
Recall that $H$ is a binary tree of qubits with $(n+1)$-layer. The $l$th layer of $H$ is denoted as $H_{l}$, and the $j$th node of $H_l$ is denoted as $H_{l,j}$ ($0\leqslant l\leqslant n$, $0\leqslant j\leqslant 2^l-1$). We require another $n$ binary trees, each denoted as $V_l$ ($1\leqslant l\leqslant n$). $V_l$ has $(l+1)$ layers, and the $m$th layer of $V_l$ is denoted as $V_{l,m}$ ($0\leqslant m\leqslant l$). The $j$th node of $V_{l,m}$ is denoted as $V_{l,m,j}$ ($0\leqslant j\leqslant 2^m-1$). The leaf $V_{l,l,j}$ is connected to $H_{l,j}$. $V_{l,0}$ corresponds to $V_{l,\text{root}}$ in the main text and $V_{l,l}$ correspnds to $V_{l,\text{leaf}}$ in the main text. The total space complexity is therefore $O(N)$.  Before preparation, $H_{0,0}$ is initialized to $|1\rangle$ while all other nodes are initialized to $|0\rangle$.

For target state $\sum_{k=0}^{N-1}\alpha_k|k\rangle$, we define $b_{n,k}\equiv|a_{k}|$, $b_{l,j}\equiv\sqrt{|b_{l+1,2j}|^2+|b_{l+1,2j+1}|^2}$ for $0\leqslant l\leqslant n-1$, and $\theta_{l,j}\equiv \arccos (b_{l,2j}/b_{l-1,j})$.
It is worthy to note that from the classical description in Eq.~\eqref{eq:ab}, compiling the quantum circuit with $O(N)$ space complexity can also be realized with circuit depth $O(n)$, because one can calculate all $\theta_{l,j}$ for same $l$ with $O(1)$ circuit depth  in parallel.

 The pseudo code of arbitrary state preparation is given in Alg.~\ref{alg:ab}. We have defined 
\begin{align}\label{eq:S}
S(\theta,a,b)=\begin{pmatrix}1&0&0&0\\0&0&\sin\theta&\cos\theta&\\0&0&\cos\theta&-\sin\theta&\\0&1&0&0\end{pmatrix}
\end{align}
as two qubit gates applied to nodes $a$ and $b$ with basis $\{|00\rangle,|01\rangle,|10\rangle,|11\rangle\}$, and $\text{Ph}(\theta,a)$ as the single qubit phase gate $\begin{pmatrix}1&\\&e^{i\theta}\end{pmatrix}$
applied to node $a$ with basis $\{|0\rangle,|1\rangle\}$. NOT$(a)$ represents the NOT gate applied to $a$; SWAP$(a,b)$ represents the swap gate applied at $a$ and $b$; CNOT$(a;b)$ represents the CNOT gate with $a$ as controlled qubit and $b$ as target qubit; CCNOT$(a,b;c)$ represents the three-qubit controlled-controlled-NOT gate with $a$ and $b$ as controlled qubits and $c$ as target qubit.

The algorithm contains $5$ stages as follows.

\textit{Stage 1}. This stage corresponds to line $1$-$6$ in Alg.~\ref{alg:ab}.
Take the first and second steps (for loop in line $1$ with $l=1,2$) as examples, the quantum state for $H_0$, $H_1$ and $H_2$ follows the transformation
\begin{align}
&|1\rangle_{H_{0}}|0\rangle^{\otimes2}_{H_{1}}|0\rangle^{\otimes4}_{H_{2}}\notag\\
\mathop{\longrightarrow}\limits^{\text{}}&|1\rangle_{H_{0}}|10\rangle_{H_{1}}|0\rangle^{\otimes4}_{H_{2}}\notag\\
\mathop{\longrightarrow}\limits^{\text{}}&|1\rangle_{H_{0}}(b_{1,0}|10\rangle_{H_{1}}+b_{1,1}|01\rangle_{H_{1}})|0\rangle^{\otimes4}_{H_{2}}\notag\\
\mathop{\longrightarrow}\limits^{\text{}}&|1\rangle_{H_{0}}\left(b_{1,0}|10\rangle_{H_{1}}|1000\rangle_{H_{2}}+b_{1,1}|01\rangle_{H_{1}}|0010\rangle_{H_{2}}\right)\notag\\
\mathop{\longrightarrow}\limits^{\text{}}&|1\rangle_{H_{0}}\left(b_{1,0}|10\rangle_{H_{1}}\left(\frac{b_{2,0}}{b_{1,0}}|1000\rangle_{H_{2}}+\frac{b_{2,1}}{b_{1,0}}|0100\rangle_{H_{2}}\right)+b_{1,1}|01\rangle_{H_{1}}\left(\frac{b_{2,2}}{b_{1,1}}|0010\rangle_{H_{2}}+\frac{b_{2,3}}{b_{1,1}}|0001\rangle_{H_{2}}\right)\right)\notag\\
=&\sum_{k=0}^{2^2-1}b_{2,k}|1\rangle_{H_{0}}\bigotimes_{l=1}^{2}|(k,l)\rangle'_{H_{l}}.
\end{align}
We have defined $(k,l)\equiv k_nk_{n-1}\cdots k_{n-l+1}$, $|(k,l)\rangle'\equiv|0\rangle^{\otimes (k,l)}|1\rangle|0\rangle^{\otimes 2^l-(k,l)-1}$. We has also defined the tensor product $\bigotimes_{l=L_a}^{L_{b}}x_{l}\equiv x_{L_a}\otimes x_{L_a+1}\otimes \cdots\otimes  x_{L_{b}}$ for $L_{b}>L_a$, or $\bigotimes_{l=L_a}^{L_{b}}x_{l}\equiv x_{L_a}\otimes x_{L_a-1}\otimes \cdots\otimes  x_{L_{b}}$ for $L_{b}\leqslant L_a$.
Similarly, after the $n$th step, the quantum state of $H$ becomes
\begin{align}\label{eq:1n}
|\psi'_{\text{stage 1}}\rangle=\sum_{k=0}^{2^n-1}b_{n,k}|1\rangle_{H_{0}}\bigotimes_{l=1}^{n}|(k,l)\rangle'_{H_{l}}.
\end{align}
After implementing line $6$ of Alg.~\ref{alg:ab}, we obtain
\begin{align}\label{eq:app_out_ed}
|\psi_{\text{stage 1}}\rangle=\sum_{k=0}^{2^n-1}a_{k}|1\rangle_{H_{0}}\bigotimes_{l=1}^{n}|(k,l)\rangle'_{H_{l}},
\end{align}
which is Eq.~\eqref{eq:ab_2} in the main text. Each step of this stage have circuit depth $O(1)$, so the total circuit depth is $O(n)$.

\textit{Stage 2}. This stage corresponds to lines $7$-$9$ of Alg.~\ref{alg:ab} with PCNOT the parallel CNOT gates described in Sec.~\ref{sec:pcnot_fanout} (Alg.~\ref{alg:pcnot}). The goal of this stage is to perform the transformation $|(k,l)\rangle'_{H_l}|0\rangle_{V_{l,0}}\rightarrow|(k,l)\rangle'_{H_l}|k_{n-l+1}\rangle_{V_{l,0}}$.

There is only one activated qubit (at state $|1\rangle$) at each layer while all other qubits are at state $|0\rangle$. We suppose the activated qubit at layer $H_l$ is $H_{l,j_{\text{act}(l)}}$. For Eq.~\eqref{eq:app_out_ed}, it can be verified that if $j_{\text{act}}(l)$ is an even, we have $|k_{n-l+1}\rangle=|0\rangle$. Similarly, if $j_{\text{act}}(l)$ is an odd, we have $|k_{n-l+1}\rangle=|1\rangle$. As an example, for $|k\rangle=|1010\rangle$, we have $|(1010,1)\rangle'=|1\rangle'=|01\rangle$, $|(1010,2)\rangle'=|10\rangle'=|0010\rangle$, $|(1010,3)\rangle'=|101\rangle'=|00000100\rangle$, $|(1010,4)\rangle'=|1010\rangle'=|0000000000100000\rangle$. The activated qubit at each layer corresponds to $j_{\text{act}}(1)=1$ (odd), $j_{\text{act}}(2)=2$ (even), $j_{\text{act}}(3)=5$ (odd) and $j_{\text{act}}(4)=10$ (even), respectively. Accordingly, we have $|k_{4-1+1}\rangle=|1\rangle, |k_{4-2+1}\rangle=|0\rangle, |k_{4-3+1}\rangle=|1\rangle, |k_{4-4+1}\rangle=|0\rangle$.

Based on the observation above, one just need to flip $V_{l,0,0}$ when the activated qubit at $H_l$ corresponds to an odd $j_{\text{act}}(l)$, which is the consequence of this stage. We therefore obtain
\begin{align}
|\psi_{\text{stage 2}}\rangle=\sum_{k=0}^{2^n-1}a_{k}|1\rangle_{H_0}\bigotimes_{l=1}^{n}|(k,l)\rangle'_{H_l}|k_{n-l+1}\rangle_{V_{l,0}}.
\end{align}
 Lines $7$, $9$ of Alg.~\ref{alg:ab} has circuit depth $O(1)$, and line $7$ has circuit depth $O(n)$. So this stage has circuit depth $O(n)$.

\begin{figure}
	\centering
	\includegraphics[width=\columnwidth]{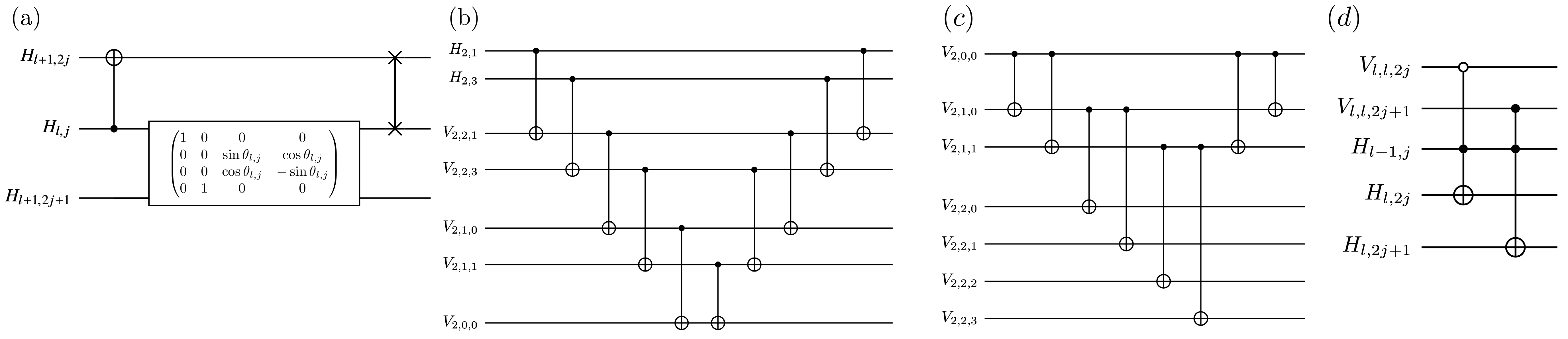}
	\caption{Illustration of quantum circuits for arbitrary quantum state preparation. (a) Operations at stage $1$ corresponding to line $2$ to $4$ of Alg.~\ref{alg:ab}. (b) Example of the transformation $|(k,2)\rangle'_{H_2}|0\rangle_{V_{2,0}}\rightarrow|(k,2)\rangle'_{H_2}|k_{n-1}\rangle_{V_{2,0}}$ at stage $2$, corresponding to line $7$ to $9$ of Alg.~\ref{alg:ab}. (c) Example of the fanout operation $|k_{n-1}\rangle_{V_{2,0}}|0\rangle^{\otimes 4}_{V_{2,2}}\rightarrow|k_ {n-1}\rangle_{V_{2,0}}|k_{n-1}\rangle^{\otimes 4}_{V_{2,2}}$ at stage 3, corresponding to line $10$ of Alg.~\ref{alg:ab}. (d) Quantum circuit for uncomputation operations at stage 4 corresponding to line $13$ to $16$ of Alg.~\ref{alg:ab}.}
	\label{fig:uc}
\end{figure}

The remaining task is to uncompute $H$. It can be observed that $H_{l,j}$ is at state $|1\rangle$, if and only if its parent is at state $|1\rangle$ and $j$ has the same parity to $k_{n-l+1}$. Therefore, by flipping $H_{l,j}$ conditioned on the states of its parent and $V_{l,0}$, one can bring $H_{l,j}$ back to state $|0\rangle$. We expect to uncompute all nodes at layer $H_l$ in a parallel way, so we need $2^l$ ``copies'' of $|k_{n-l+1}\rangle$. Based on the idea above, the uncomputation is realized with the following 3 stages.

\textit{Stage 3}. This stage corresponds to line $10$ of Alg.~\ref{alg:ab} (see also Alg.~\ref{alg:fan} in Sec.~\ref{sec:pcnot_fanout}). The fanout operation ``copy'' the state at the root of $V_l$ to the leaves of $V_l$. According to Eq.~\eqref{eq:fan1}, the quantum state of $H$, $V_{l,0}$ and $V_{l,l}$ for all $0\leqslant j\leqslant 2^l-1$ becomes
\begin{align}\label{eq:vf}
|\psi_{\text{stage 3}}\rangle=\sum_{k=0}^{2^n-1}a_{k}|1\rangle_{H_0}\bigotimes_{l=1}^{n}|(k,l)\rangle'_{H_l}|k_{n-l+1}\rangle_{V_{l,0}}|k_{n-l+1}\rangle^{\otimes 2^l}_{V_{l,l}}.
\end{align}
Because the fanout operation for each $l$ can be implemented in parallel, the circuit depth of this stage is $O(n)$.

\textit{Stage 4}. This stage corresponds to line $11$-$17$ of Alg.~\ref{alg:ab}, and our goal is to uncompute $H$. In each basis of Eq.~\eqref{eq:vf}, we have
\begin{align}
H_{l,2j} \;\text{is activated} &\Longleftrightarrow \;V_{l,l,2j} \;\text{is not activated, and} \;H_{l-1,j} \;\text{is activated}\\
H_{l,2j+1} \;\text{is activated} &\Longleftrightarrow \;\text{both}\; V_{l,l,2j+1}\; \text{and} \;H_{l-1,j} \;\text{are activated}
\end{align}
So line $12$ and $13$ uncomputes $H_{l,2j}$ and $H_{l,2j+1}$ respectively.
With total circuit depth $O(n)$, the quantum state after this stage becomes
\begin{align}\label{eq:??}
|\psi_{\text{stage 4}}\rangle=\sum_{k=0}^{2^n-1}a_{k}|1\rangle_{H_0}\bigotimes_{l=1}^{n}|0\rangle^{\otimes 2^l}_{H_l}|k_{n-l+1}\rangle_{V_{l,0}}|k_{n-l+1}\rangle^{\otimes 2^l}_{V_{l,l}}.
\end{align}

\textit{Stage 5}. This stage corresponds to line $18$ of Alg.~\ref{alg:ab}, after which we obtain
\begin{align}\label{eq:???}
|\psi'_{\text{stage 5}}\rangle=\sum_{k=0}^{2^n-1}a_{k}|1\rangle_{H_0}\bigotimes_{l=1}^{n}|0\rangle^{\otimes 2^l}_{H_l}|k_{n-l+1}\rangle_{V_{l,0}}|0\rangle^{\otimes 2^l}_{V_{l,l}}.
\end{align}
By tracing out all nodes except for $V_{l,0}$ for $l\in\{1,\cdots,n\}$, we obtain
\begin{align}\label{eq:targ_sup}
|\psi_{\text{stage 5}}\rangle=\sum_{k=0}^{2^n-1}\alpha_k\bigotimes_{l=1}^{n}|k_{n-l+1}\rangle_{V_{l,0}}.
\end{align}
Because $|k\rangle\equiv|k_nk_{n-1},\cdots k_1\rangle$,  Eq.~\eqref{eq:targ_sup} is equivalent to the target state $|\psi\rangle$ with $V_{l,0}$ being the $(n-l+1)$th digit.

To sum up, the target state is prepared with totally $\Theta(n)$ circuit depth, saturating the lower bound. Note that for reaching the circuit depth $\Theta(n)$, at least $\Omega(N/n)$ ancillary qubits are required. So the space complexity $O(N)$ of our method is nearly optimal.

\subsection{Parallel CNOT gate and fanout operation based on binary tree}\label{sec:pcnot_fanout}
Given an $(L+1)$-layer binary tree $T$ (each node represents a qubit), we denote its $l$th ($0\leqslant l\leqslant L$) layer as $T_l$, and the $j$th node ($0\leqslant j\leqslant 2^l-1$) at the $l$th layer as $T_{l,j}$. The parallel CNOT (PCNOT) gate is described in Alg.~\ref{alg:pcnot}. Note that the operations at each line can be realized with $O(1)$ circuit depth. To see how it works, we consider a subspace with all $T_{l,j}$ for $1\leqslant l\leqslant L-1$ at $|0\rangle$, and there is either zero or single activated qubit at $T_L$. The available basis for $T_L$ is $\mathcal{L}=\left\{|\text{vac}\rangle',|0\rangle',|1\rangle' \cdots |2^L-1\rangle'\right\}$, where $|\text{vac}\rangle'\equiv|0\rangle^{\otimes 2^L}$ and $|j\rangle'\equiv|0\rangle^{\otimes j}|1\rangle|0\rangle^{\otimes 2^L-j-1}$, and the basis for $T_{0,0}$ is $\{|0\rangle,|1\rangle\}$. We have
\begin{align}
&\text{PCNOT}(T)|z\rangle_{T_0}|j\rangle'_{T_L}=|\overline z\rangle_{T_0}|j\rangle'_{T_L} \quad\text{for}\; |j\rangle'_{T_L}\neq|\text{vac}\rangle_{T_L}, \\
&\text{PCNOT}(T)|z\rangle_{T_0}|j\rangle'_{T_L}=|z\rangle_{T_0}|j\rangle'_{T_L} \quad\text{for}\; |j\rangle'_{T_L}=|\text{vac}\rangle_{T_L}.
\end{align}
In other words, we flip the root if there is an activated qubit at the leaf layer.

The fanout operation for $T$ is described in Alg.~\ref{alg:fan}. In the same subspace above for PCNOT, we have
\begin{align}
&\text{Fanout}(T)|z\rangle_{T_0}|\text{vac}\rangle_{T_L}=|z\rangle_{T_0}|z\rangle^{\otimes 2^L}_{T_L},\label{eq:fan1}\\
&\text{Fanout}(T)|z\rangle_{T_0}|z\rangle^{\otimes 2^L}_{T_L}=|z\rangle_{T_0}|\text{vac}\rangle_{T_L}.\label{eq:fan2}
\end{align}
Each step in Alg.~\ref{alg:pcnot} and Alg.~\ref{alg:fan} can be realized with circuit depth $O(1)$, so both operations have total circuit depth $O(L)$.

\begin{algorithm} [H]
\caption{: Arbitrary quantum state preparation $\\$Input: rotation angles $\theta_{l,j}$ for $1\leqslant l\leqslant n$ and $0\leqslant j\leqslant 2^l-1$; phase angles $\arg(a_k)$ for $0\leqslant k\leqslant 2^n-1$
 $\\$ Require:  Binary trees $H$ and $V_l$ for all $1\leqslant l\leqslant n$ with all nodes initialized as $|0\rangle$, except for $H_{0,0}$ initialized as $|1\rangle$
 $\\$ Output: quantum state described in Eq.~\eqref{eq:ab} with $V_{l,0}$ being the $(n-l+1)$th digit}
\label{alg:ab}
\begin{algorithmic}[1]

\STATE \textbf{for}  $l=1$ to $n$:
\STATE \quad $\text{CNOT}\left(H_{l-1,j};H_{l,2j}\right)$ for all $0\leqslant j\leqslant 2^{l-1}-1$ 
\STATE \quad $S(\theta_{l,j}, H_{l-1,j},H_{l,2j+1})$ for all $0\leqslant j\leqslant 2^{l-1}-1$ 
\STATE \quad $\text{SWAP}\left(H_{l,2j},H_{l-1,j}\right)$ for all $0\leqslant j\leqslant 2^{l-1}-1$ 
\STATE \textbf{end}
\STATE $\text{Ph}(\arg(a_j), H_{n,j})$ for all $0\leqslant j\leqslant 2^n-1$ 
\STATE $\text{CNOT}(H_{l,j};V_{l,l,j})$ for all $1\leqslant l\leqslant n$ and $j\in\{1,3,5 \cdots 2^l-1\}$ 
\STATE $\text{PCNOT}(V_l)$ for all $1\leqslant l\leqslant n$ 
\STATE $\text{CNOT}(H_{l,j};V_{l,l,j})$ for all $1\leqslant l\leqslant n$ and $j\in\{1,3,5 \cdots 2^l-1\}$ 
\STATE $\text{Fanout}(V_l)$ for all $1\leqslant l\leqslant n$ 
\STATE \textbf{for}  $l'=1$ to $n$:
\STATE  \quad Let $l=n-l'+1$

\STATE \quad $\text{NOT}(V_{l,l,2j})$ for all $j\in\{0,1,2 \cdots 2^{l-1}-1\}$ 
\STATE \quad $\text{CCNOT}(V_{l,l,2j},H_{l-1,j};H_{l,2j})$ for all $j\in\{0,1,2 \cdots 2^{l-1}-1\}$ 
\STATE \quad $\text{NOT}(V_{l,l,2j})$ for all $j\in\{0,1,2 \cdots 2^{l-1}-1\}$ 
\STATE \quad $\text{CCNOT}(V_{l,l,2j+1},H_{l-1,j};H_{l,2j+1})$ for all $j\in\{0,1,2 \cdots 2^{l-1}-1\}$ 
\STATE \textbf{end}
\STATE $\text{Fanout}(V_l)$ for all $1\leqslant l\leqslant n$ 
\end{algorithmic}
\end{algorithm}

\begin{algorithm} [H]
\caption{: $\text{PCNOT}(T)$  }
\label{alg:pcnot}
\begin{algorithmic}[1]

\STATE \textbf{for}  $l'=0$ to $L-1$:
\STATE \quad Let $l=L-l'$

\STATE \quad  $\text{CNOT}\left(T_{l,j};T_{l-1,j/2}\right)$ for all $j\in\{0,2,4 \cdots 2^l-2\}$

\STATE \quad $\text{CNOT}\left(T_{l,j};T_{l-1,(j-1)/2}\right)$ for all $j\in\{1,3,5 \cdots 2^l-1\}$

\STATE \textbf{end}

\STATE \textbf{for}  $l'=0$ to $L-2$:

\STATE \quad Let $l=L-l'$

\STATE \quad  $\text{CNOT}\left(T_{l,j};T_{l-1,j/2}\right)$ for all $j\in\{0,2,4 \cdots 2^l-2\}$

\STATE \quad $\text{CNOT}\left(T_{l,j};T_{l-1,(j-1)/2}\right)$ for all $j\in\{1,3,5 \cdots 2^l-1\}$
\STATE \textbf{end}

\end{algorithmic}
\end{algorithm}

\begin{algorithm} [H]
\caption{: $\text{Fanout}(T)$ }
\label{alg:fan}
\begin{algorithmic}[1]

\STATE \textbf{for}  $l=0$ to $L-1$:
\STATE \quad $\text{CNOT}\left(T_{l,j};T_{l+1,2j}\right)$ for all $0\leqslant j\leqslant 2^l-1$
\STATE \quad $\text{CNOT}\left(T_{l,j};T_{l+1,2j+1}\right)$ for all $0\leqslant j\leqslant 2^l-1$
\STATE \textbf{end}

\STATE \textbf{for}  $l=0$ to $L-2$:
\STATE \quad $\text{CNOT}\left(T_{l,j};T_{l+1,2j}\right)$ for all $0\leqslant j\leqslant 2^l-1$
\STATE \quad $\text{CNOT}\left(T_{l,j};T_{l+1,2j+1}\right)$ for all $0\leqslant j\leqslant 2^l-1$
\STATE \textbf{end}

\end{algorithmic}
\end{algorithm}

\section{circuit depth lower bound for sparse state preparation}
\begin{lemma}\label{lm:opt_s}
With only single- and two-qubit gates and arbitrary amount of ancillary qubits, the minimum circuit depth of preparing arbitrary $d$-sparse ($ d\geqslant 2$) quantum states is $\Omega(\log(nd))$.
\end{lemma}

\begin{proof}
Firstly, we consider the preparation of the GHZ state $1/\sqrt{2}\left(|0\rangle^{\otimes n}+|1\rangle^{\otimes n}\right)$, which can be considered as a $d$-sparse state with any $d\geqslant2$. For initial state $|0\rangle^{\otimes n}$, there is no correlation between each qubit, while in the target GHZ state, each qubit is correlated to all other qubits. To establish the correlations between one qubit to all other $n-1$ qubits with single- and two-qubit gates, $\Omega(\log n)$ circuit depth is necessary.

Secondly, when preparing an $n$-qubit state with a $L$-depth quantum circuit using only single- and two-qubit gates, there are at most $O(n\cdot 2^L)$ effective parameters (Ref~\cite{Sun.21_new}, Lemma 28). On the other hand, specifying a $d$-sparse ($2\leqslant d\leqslant2^n$) quantum state requires no less than $\Omega(d)$ parameters. Therefore, to prepare arbitrary $d$-sparse quantum states, the circuit depth is lower bounded by $\Omega(\log(d/n))=\Omega(\log d-\log n)$.

Combining results above, the circuit lower bound is $\Omega(\log(nd))$.
\end{proof}
The derivation above also indicates that to reach the circuit depth $\Theta(\log(nd))$, at least $\Omega(d-n)$ ancillary qubits are required. If we further consider the resource for storing the values and indexes of all nonzero entries, the total space complexity is lower bounded by $\Omega(nd)$.

\section{product unitary memory (PUM) and sparse Boolean memory (SBM)}\label{sec:us_sup}
In this section,  we provide details about the data structure introduced in Lemma~\ref{lm:1} (PUM) and Lemma~\ref{lm:2} (SBM). We begin with single-word cases (word register has only one qubit) in Sec.~\ref{sec:pum_sig} (PUM) and Sec.~\ref{sec:sbm_sig} (SBM). Then, in Sec.~\ref{sec:multi}, we discuss how to generalize PUM and SBM to multi-word cases (word register has more than one qubit).

\subsection{Single-word PUM}\label{sec:pum_sig}

In this subsection, our goal is to perform the transformation
\begin{align}\label{eq:pum_sig}
\text{select}(\hat U)\sum_{k=0}^{d-1}\sum_{z=0,1}\psi_{k,z}|k\rangle|z\rangle=\sum_{k=0}^{d-1}\sum_{z=0,1}\psi_{k,z}|k\rangle \hat U(k)|z\rangle,
\end{align}
where $|k\rangle$ is the $\tilde n$-qubit basis ($d=2^{\tilde n}$) of the index register and $|z\rangle$ is the basis of the single-qubit word register. Our method can be considered as a generalization of the binary QRAM protocol~\cite{Giovannetti.08,Giovannetti.08_2,Hann.19,Hann.21} to continuous data. We will introduce the main idea of our protocol in Sec.~\ref{sec:pum_sig_m} and then provide the formal description in Sec.~\ref{sec:pum_sig_f}.

\subsubsection{Main idea}\label{sec:pum_sig_m}
The element of the layout is the \textit{router} containing four qubits. If we denote the incident qubit, routing qubit, left output qubit and right output qubit as $q_{\text{incident}},q_{\text{route}},q_{\text{left}}$ and $q_{\text{right}}$ respectively, the \textit{routing} operation is defined as
\begin{align}\label{eq:routing}
\text{ROUTE}=\text{CSWAP}(q_{\text{route}};q_{\text{incident}},q_{\text{right}})\text{NOT}(q_{\text{route}})\text{CSWAP}(q_{\text{route}};q_{\text{incident}},q_{\text{left}})\text{NOT}(q_{\text{route}}).
\end{align}
By applying the routing operation, the state of $q_{\text{incident}}$ is transferred to $q_{\text{left}}$ ($q_{\text{right}}$), if the routing qubit is at state $|0\rangle$ ($|1\rangle$).

The layout of single-word PUM is shown in Fig.~\ref{fig:qram}. Except for index register and word register, we also introduce another pointer qubit initialized as $|1\rangle$. The layout contains ($2d-1$) routers arranged as a $\tilde n$-layer binary tree. Each left (right) output qubit serves as the incident qubit of its left (right) child. Each left and right output qubit of the node at the leaf layer connects to a \textit{memory cell}. Each memory cell contains a pointer-holder qubit and a word-holder qubit initialized as $|0\rangle$. The $k$th memory cell stores a single qubit unitary $U(k)$ applied at the word-holder qubit and controlled by the pointer-holder qubit.  

\begin{figure}
	\centering
	\includegraphics[width=.75\columnwidth]{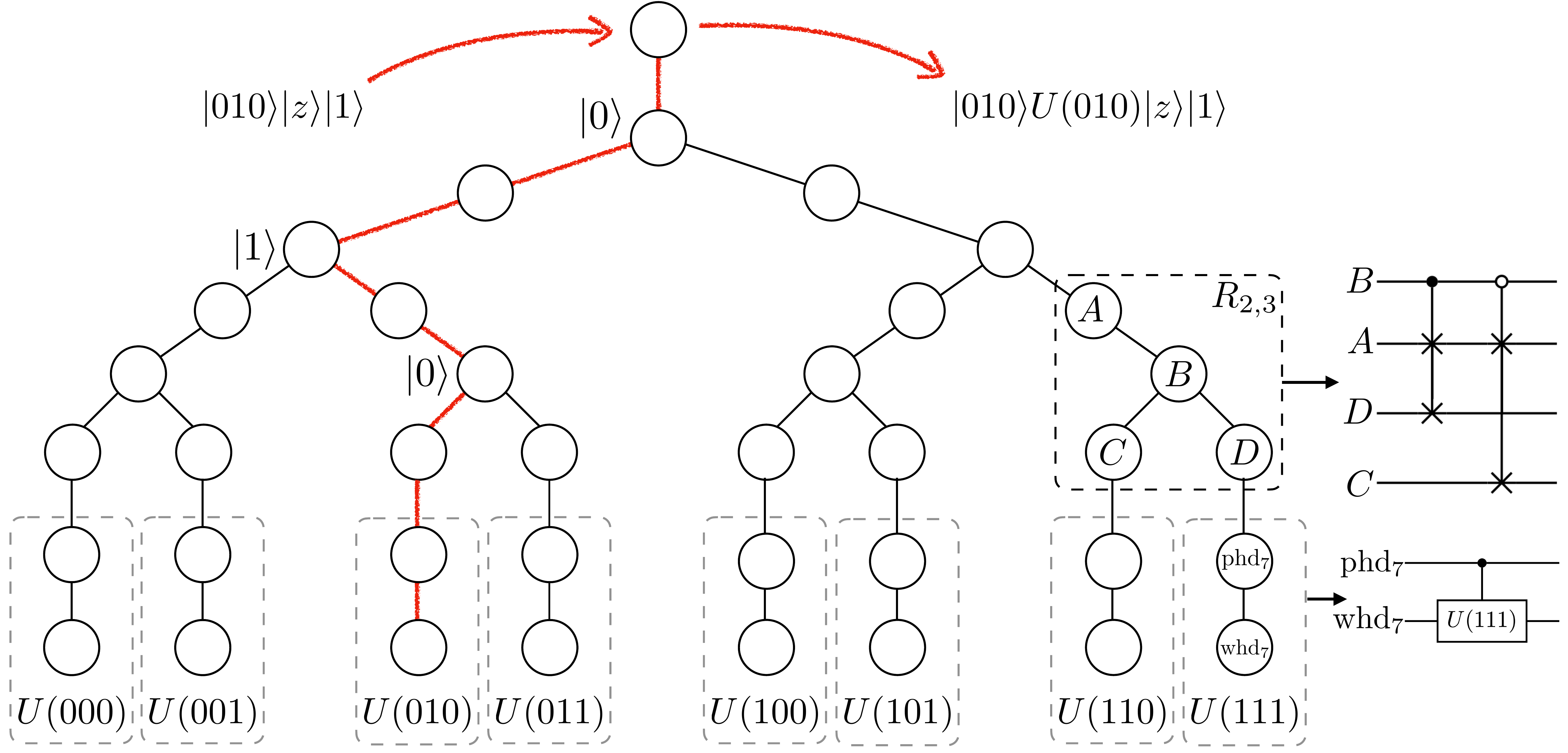}
	\caption{Sketch of single-word PUM. Main panel is the qubit layout of the PUM with $d=8$ ($\tilde n=3$). Each circle represents a qubit, and solid line represents the connection between qubits. Dashed boxes with $4$ qubits inside represents the router with label $R_{2,3}$. $A$, $B$, $C$ and $D$ represents the incident qubit, routing qubit, left output qubit and right output qubit respectively. The corresponding circuit for routing operation defined in Eq.~\eqref{eq:routing} is also given in its right side. Dashed boxes with $2$ qubits inside represents the memory cells.  An example is given for the last memory cell storing $U(111)$.  phd$_7$ represents pointer-holder qubit and whd$_7$ represents word-holder qubit. In the query stage, a single qubit gate $U(111)$ is applied on the word-holder qubit whd$_7$ controlled by the pointer-holder qubit phd$_7$ (line 2 of Alg.~\ref{alg:route_full}). If the index register, word register and pointer qubit are at state $|010\rangle$, $|z\rangle$ and $|1\rangle$ respectively, state $|z\rangle|1\rangle$ is routed in the word holder and pointer holder of the third memory cell. After route out, the output state becomes $|010\rangle U(010)|z\rangle|1\rangle$. The path for route-in and route-out is represented by red line.}
	\label{fig:qram}
\end{figure}

Our method has three stages. In the route-in stage, the state of index register, word register and pointer qubit are sent into the binary tree sequentially. ROUTE defined by Eq.~\eqref{eq:routing} transfers the state of the incident qubit of a node to the incident qubit of its left (right) child, conditioned on the routing qubit is at state $|0\rangle$ ($|1\rangle$). The state of the $l$th index qubit is finally transferred to one of the routing qubits at the $l$th layer. The states of the word register and the pointer qubit travels to the word-holder and pointer-holder qubits of the $k$th memory cell, conditioned on the intup index register at state $|k\rangle$.

In the query stage, at the $k$th memory cell, we implement controlled-$\hat U(k)$ with pointer-holder as controlled qubit and the word-holder as target qubit. This gives either the transformation $|z\rangle|1\rangle\rightarrow\hat U(k)|z\rangle|1\rangle$ or $|00\rangle\rightarrow|00\rangle$.

 In the route-out stage, we perform the inverse of the route-in stage. The index register, word register and the pointer qubit are transferred to the state $\sum_{k=0}^{d-1}\psi_{k,z}|k\rangle U(k)|z\rangle|1\rangle$, while all other qubits are uncomputed. By tracing out the pointer qubit, we obtain the target state.

An example is given in Fig.~\ref{fig:qram}. The index and word register are initially at $|010\rangle$ and $|z\rangle$ respectively. Together with the pointer qubit, the input state is  $|010\rangle|z\rangle|1\rangle$. After the route-in stage, the quantum state $|z\rangle|1\rangle$ of the word register and pointer qubit are sent into the memory cell storing $U(010)$. After the query stage, the state at the memory cell storing $U(010)$ becomes $U(010)|z\rangle|1\rangle$ while all other memory cells remain unchanged (at state $|0\rangle|0\rangle$). After the route-out stage, memory cell and the binary tree are uncomputed, while the output state of the index register, word register and pointer qubit becomes $|010\rangle U(010)|z\rangle|1\rangle$ as expected. 

\subsubsection{Formal description}\label{sec:pum_sig_f}
We denote the $j$th router at the $l$th layer as $R_{l,j}$, and denote the pointer-holder qubit and word-holder qubit at the $k$th memory cell as phd$_k$ and whd$_k$, respectively. For $0\leqslant k\leqslant d-1$, we also denote the the left output qubit of the $R_{\tilde{n},k/2}$ and the right output qubit of the $R_{\tilde{n},(k-1)/2}$  as $O_{k}$. O$_k$ can also be considered as the $k$th leaf node of the router tree.

We define several operations required by our protocol as follow:

\begin{itemize}
\item Input$(m)$: swap gate between the $m$th qubit at the index register and the incident qubit of $R_{0,0}$
\item Input(pointer): swap gate between pointer qubit, and the incident qubit of $R_{0,0}$
\item Input(word): swap gate between the qubit at word register and the incident qubit of $R_{0,0}$
\item Route($l$): routing operation defined in Eq.~\eqref{eq:routing} for all routers $R_{l,j}$ with $j\in\{0,2^l-1\}$
\item Set($l$): swap gate between the incident qubit and routing qubit of $R_{l,j}$ for all $j\in\{0,1,\cdots, 2^l-1\}$
\item Output(pointer): swap gate between O$_k$ and the pointer holder $\text{phd}_k$ for all $k\in\{0,1 \cdots d-1\}$
\item Output(word): swap gate between O$_k$ and the word-holder qubit $\text{whd}_k$ for all $k\in\{0,1 \cdots d-1\}$~\cite{fn} 
\item CU(phd$_k$;whd$_k$): Controlled-$U(k)$ with $\text{phd}_k$ as controlled qubit and $\text{whd}_k$ as target qubit for all $k\in\{0,1 \cdots d-1\}$
\end{itemize}
Note that each operation above can be implemented with circuit depth $O(1)$.
The pseudo code for the full process is given in Alg.~\ref{alg:route_full}, and the route-in stage is given in Alg.~\ref{alg:route_route}, where we have assumed that $\tilde n/6$ is an integer. We note that the route-in protocol is first given in the Appendix A of ~\cite{Hann.21}.

\begin{algorithm} [H]
\caption{: select($\hat U$) for single-word case}
\label{alg:route_full}
\begin{algorithmic}[1]

\STATE Route-in
\STATE CU$\left(\text{phd}_k; \text{whd}_k\right)$ for all $k\in\{0,1\cdots d-1\}$
\STATE inverse of Route-in

\end{algorithmic}
\end{algorithm}

\begin{algorithm} [H]
\caption{: Route-in (Appendix A of ~\cite{Hann.21})}
\label{alg:route_route}
\begin{algorithmic}[1]

\STATE \textbf{for} $L=1$ to $\tilde n/3$

\STATE\quad Input($3L-2$); Route($2l$) for all $1\leqslant l\leqslant (L-1)$

\STATE\quad  Route($2l-1$) for all $1\leqslant l\leqslant (L-1)$; Set($2L-1$)

\STATE\quad  Input($3L-1$); Route($2l$) for all $1\leqslant l\leqslant (L-1)$

\STATE\quad  Route($2l-1$) for all $1\leqslant l\leqslant L$

\STATE\quad  Input($3L$); Route($2l$) for all $1\leqslant l\leqslant (L-1)$; Set($2L$)

\STATE\quad  Route($2l-1$) for all $1\leqslant l\leqslant L$

\STATE \textbf{end}

\STATE Input(pointer); Route($2l$) for all $1\leqslant l\leqslant \tilde n/3$

\STATE Route($2l-1$) for all $1\leqslant l\leqslant \tilde n/3$; Set($2\tilde n/3+1$)

\STATE Input(word); Route($2l$) for all $1\leqslant l\leqslant \tilde n/3$

\STATE \textbf{for} $L=\tilde n/3+1$ to $\tilde n/2-1$

\STATE\quad  Route($2l-1$) for all $(3L-\tilde n-2)\leqslant l\leqslant L$

\STATE\quad  Route($2l$) for all $(3L-\tilde n-2)\leqslant l\leqslant (L-1)$; Set($2L$)

\STATE\quad  Route($2l-1$) for all $(3L-\tilde n-1)\leqslant l\leqslant L$

\STATE\quad  Route($2l$) for all $(3L-\tilde n-1)\leqslant l\leqslant L$

\STATE\quad Route($2l-1$) for all $(3L-\tilde n)\leqslant l\leqslant L$; Set($2L+1$);

\STATE\quad  Route($2l$) for all $(3L-\tilde n)\leqslant l\leqslant L$;

\STATE \textbf{end}

\STATE Route($2l-1$) for all $(\tilde n/2-2)\leqslant l\leqslant \tilde n/2$

\STATE Route($2l$) for all $(\tilde n/2-3)\leqslant l\leqslant (\tilde n/2-1)$; Set($\tilde n$)

\STATE Route($2l-1$) for all $(\tilde n/2-1)\leqslant l\leqslant \tilde n/2$

\STATE Route($\tilde n-2$); Output(pointer)

\STATE Route($\tilde n-1$)

\STATE Output(word)

\end{algorithmic}
\end{algorithm}

\subsection{Single-word sparse Boolean memory (SBM)}\label{sec:sbm_sig}
In this subsection, our goal is to perform the transformation
\begin{align}\label{eq:sbm_in}
\sum_{k=0}^{2^n-1}\sum_{z=0,1}\psi_k|k\rangle|z\rangle\longrightarrow\sum_{k=0}^{2^n-1}\sum_{z=0,1}\psi_k|k\rangle|z\oplus f(k)\rangle,
\end{align}
where $|k\rangle$ is the state of the $n$-qubit index register, and $|z\rangle$ is the state of single-qubit word register, and $f(k)$ is an $s$-sparse Boolean function. Without loss of generality, we assume that both $L_{s}\equiv\log_2s$ and $L_{n}\equiv\log_2n$ are integers. If $L_n$ is not, we add extra qubits with state $|0\rangle$ to the index register until the assumption is satisfied.

\begin{figure}
	\centering
	\includegraphics[width=.82\columnwidth]{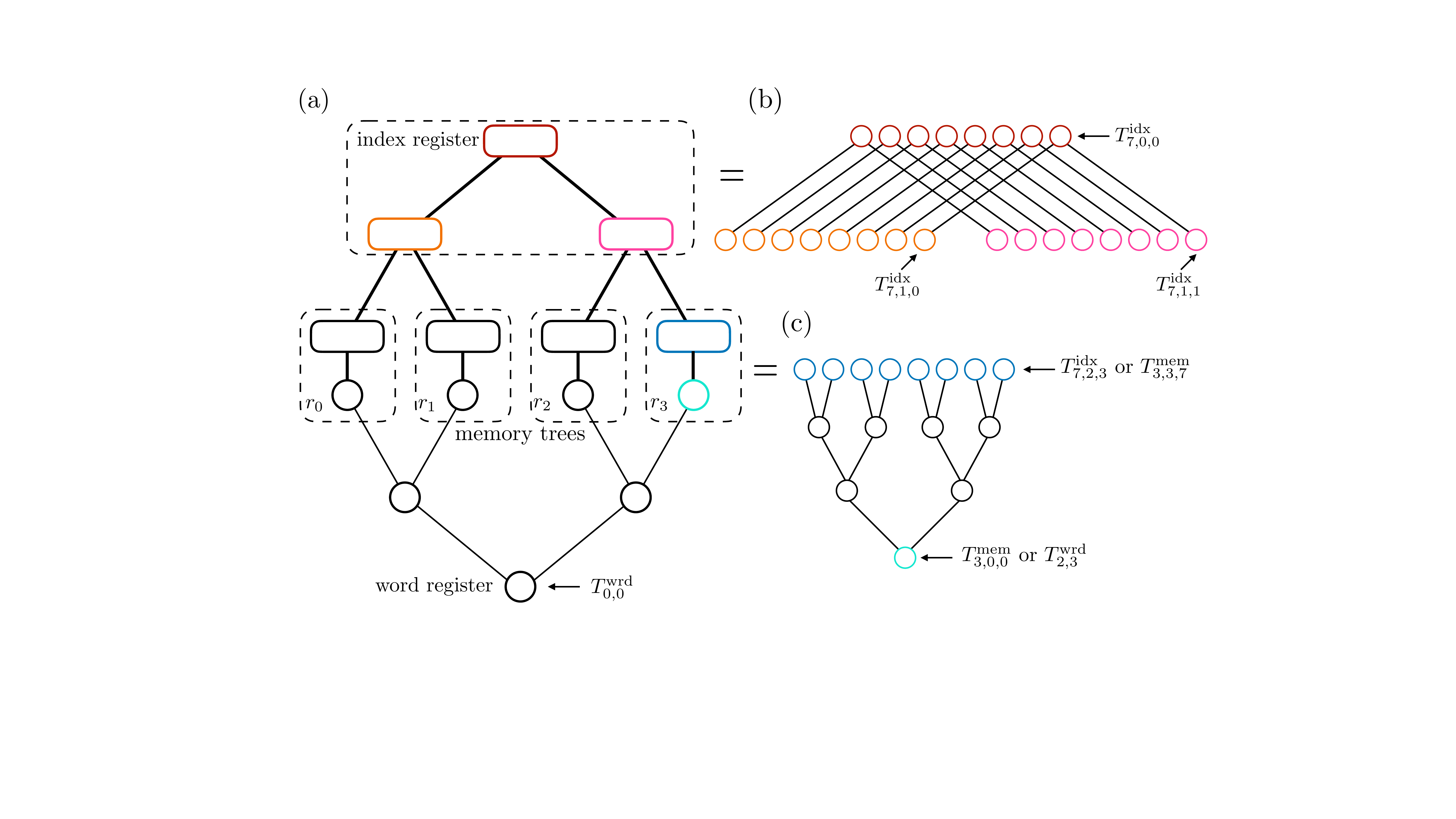}
	\caption{Qubit layout of single-word SBM with $n=8$ and $s=4$. Each circle represents a qubit. Entangling gates are applied at qubit pairs connected by light lines. Rounded rectangle represents a cluster of qubits, and bold line represents a cluster of light lines. (a) general layout of SBM. Index register corresponds to the $8$-qubit system denoted by brown rounded rectangle (or equivalently, brown circles in (b)). Word register is the single qubit with label  $T^{\text{wrd}}_{0,0}$ at the bottom. The layout contains $s=4$ memory trees. Each memory tree is denoted by a bashed box with a rectangle and a circle inside. (b) Zoom-in of the first and second layer of $T^{\text{idx}}_{l}$ for $0\leqslant l\leqslant7$. The labels of the $8$-qubit index register are $T^{\text{idx}}_{0,0,0}, T^{\text{idx}}_{1,0,0},\cdots, T^{\text{idx}}_{7,0,0}$ respectively. (c) Zoom-in of the memory cell storing $r_3$. For the sake of brevity, the layer $T^{\text{mem}}_{1}$ and $T^{\text{mem}}_{2}$, corresponding to black circles have not been shown in (a). Index holder corresponds to the layer with blue color, i.e. qubits at layer $T^{\text{mem}}_{3,3}$. Word holder corresponds to the qubit with cyan color, i.e., $T^{\text{cell}}_{3,0,0}$. See Sec.~\ref{sec:sbm_sig} for more details.}
	\label{fig:sbm1}
\end{figure}
\begin{figure}
	\centering
	\includegraphics[width=.6\columnwidth]{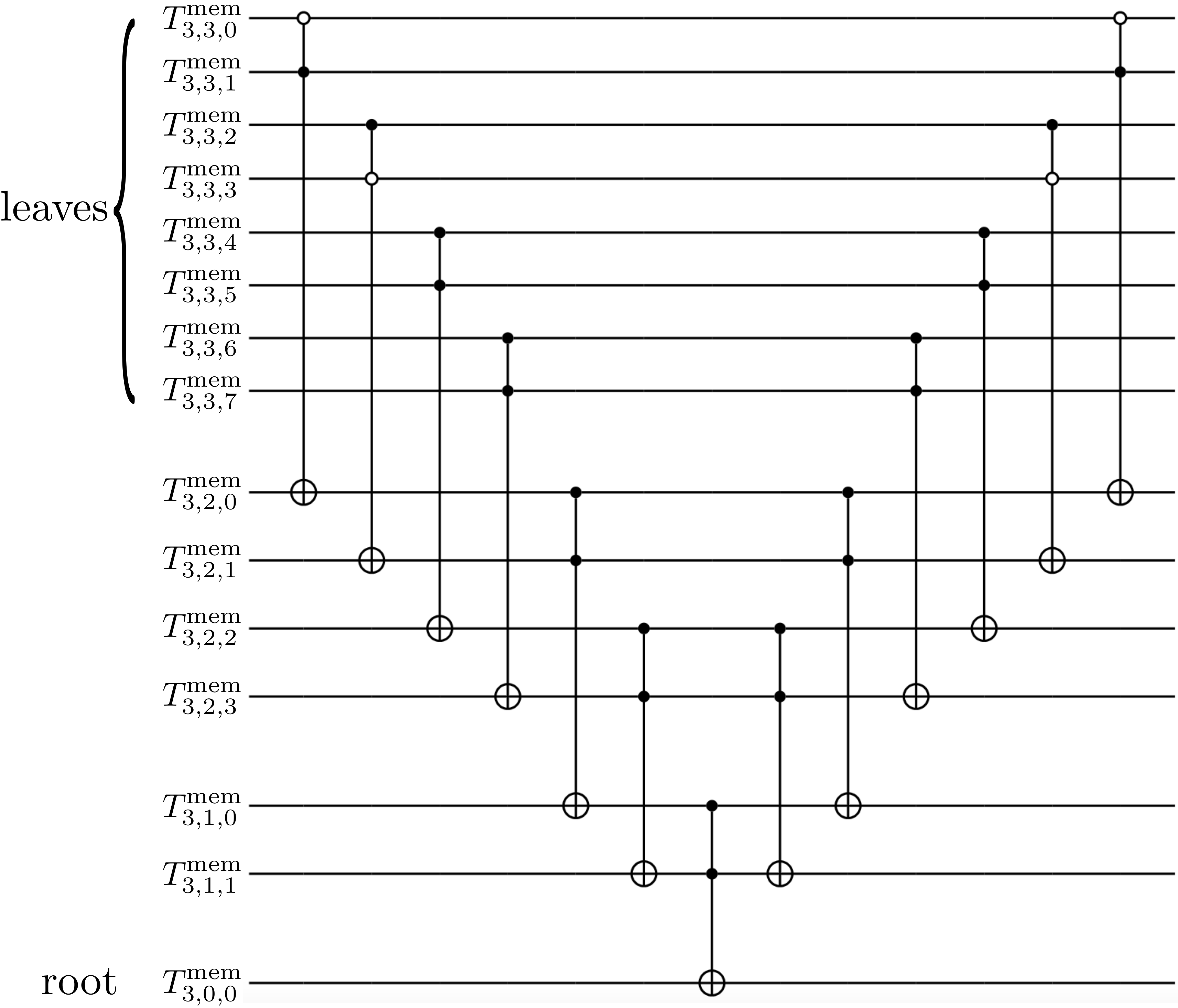}
	\caption{Quantum circuit of Toffoli$(r_3,T^{\text{mem}}_{3})$ at Alg.~\ref{alg:tof}, where we have taken $r_3=01101111$ as an example. The circuit performs the transformation $|k\rangle_{T^{\text{mem}}_{3,3}}|z\rangle_{T^{\text{mem}}_{3,0}}\rightarrow|k\rangle_{T^{\text{mem}}_{3,3}}|z\oplus (k=r_3)\rangle_{T^{\text{mem}}_{3,0}}$. }	\label{fig:sbm2}
\end{figure}

The qubit layout is illustrated in Fig.~\ref{fig:sbm1}. For each $l\in\{0,1\cdots n-1\}$, we introduce a $L_s$-layer binary tree (index tree) denoted as $T^{\text{idx}}_{l}$. We denote the $m$th layer of $T^{\text{idx}}_{l}$ as $T^{\text{idx}}_{l,m}$, and the $w$th node of $T^{\text{idx}}_{l,m}$ as $T^{\text{idx}}_{l,m,w}$. The $(l+1)$th qubit of index register serves as the root $T^{\text{idx}}_{l,0,0}$.
Let $\mathcal{S}_f=\{r_0,r_1,\cdots, r_{s-1}\}$ be a set containing all indexes satisfying $f(r_w)=1$. Our structure has totally $s$ memory cells. The $w$th cell contains a $L_n$-layer binary tree (memory tree) $T^{\text{mem}}_w$. Similarly, the $m$th layer of $T^{\text{mem}}_w$ is $T^{\text{mem}}_{w,m}$, and the $j$th node of $T^{\text{mem}}_{w,m}$ is $T^{\text{mem}}_{w,m,j}$. 
 The memory trees and index trees share their leaves. More specifically, we have $T^{\text{mem}}_{w,L_n,l}=T^{\text{idx}}_{l,L_s,w}$. We then introduce another $L_s$-layer binary tree (word tree) $T^{\text{wrd}}$. Its $w$th node at the $m$th layer is denoted as $T^{\text{wrd}}_{m,w}$. It connects the roots of all memory trees to the word register, i.e. $T^{\text{wrd}}_{L_s,w}=T^{\text{mem}}_{w,0,0}$ and $T^{\text{wrd}}_{0,0}$ is the word register.

The algorithm is shown in Alg.~\ref{alg:sbm}.
The first step is line $1$ of Alg.~\ref{alg:sbm}, which fanouts the state of index register $|k\rangle$ to the leaf layer of each memory tree. The quantum state becomes
\begin{align}
|\psi_{\text{step 1}}^{\text{sbm}}\rangle=\sum_{k=0}^{2^n-1}\sum_{z=0,1}\psi_{k,z}|k\rangle_{\text{ir}}|z\rangle_{\text{wr}}\bigotimes_{w=0}^{s-1}|k\rangle_{T_{w,L_{n}}^{\text{mem}}}|0\rangle_{T_{w,0}^{\text{mem}}}.
\end{align}
Here ``ir'' and ``wr'' represent the index register and word registers respectively.

The second step is line $2$ of Alg.~\ref{alg:sbm}. We apply the $n$-Toffoli gates defined in Alg.~\ref{alg:tof} (see also Fig.~\ref{fig:sbm2}). The transformation $|k\rangle_{T_{w,L_{n}}^{\text{mem}}}|0\rangle_{T_{w,0}^{\text{mem}}}\rightarrow|k\rangle_{T_{w,L_{n}}^{\text{mem}}}| k=r_w\rangle_{T_{w,0}^{\text{mem}}}$ is performed in each memory tree. So the state becomes
 \begin{align}
|\psi_{\text{step 2}}^{\text{sbm}}\rangle=\sum_{k=0}^{2^n-1}\sum_{z=0,1}\psi_{k,z}|k\rangle_{\text{ir}}|z\rangle_{\text{wr}}\bigotimes_{w=0}^{s-1}|k\rangle_{T_{w,L_{n}}^{\text{mem}}}|k=r_w\rangle_{T_{w,0}^{\text{mem}}}.
\end{align}
The third step is line $3$ of Alg.~\ref{alg:sbm}. There is at most one activated qubit at the leaf layer of $T^{\text{wrd}}$, so the parallel CNOT gate(defined in Alg.~\ref{alg:pcnot}) performs the transformation $|z\rangle_{\text{wr}}\rightarrow|z \oplus \prod_{w=0}^{s-1}(k=r_w)\rangle_{\text{wr}}=|z \oplus f(k)\rangle_{\text{wr}}$, and the quantum state becomes
 \begin{align}
|\psi_{\text{step 3}}^{\text{sbm}}\rangle=\sum_{k=0}^{2^n-1}\sum_{z=0,1}\psi_{k,z}|k\rangle_{\text{ir}}|z\oplus f(k)\rangle_{\text{wr}}\bigotimes_{w=0}^{s-1}|k\rangle_{T_{w,L_{n}}^{\text{mem}}}|k=r_w\rangle_{T_{w,0}^{\text{mem}}}.
\end{align}
The fourth step is line 4,5 of Alg.~\ref{alg:sbm}, aiming at uncomputing $T^\text{cell}_{w}$. After this step, the final state becomes
 \begin{align}
|\psi_{\text{step 4}}^{\text{sbm}}\rangle=\sum_{k=0}^{2^n-1}\sum_{z=0,1}\psi_{k,z}|k\rangle_{\text{ir}}|z\oplus f(k)\rangle_{\text{wr}}\bigotimes_{w=0}^{s-1}|0{\rangle^{\otimes n}}_{T_{w,L_{n}}^{\text{mem}}}|0\rangle_{T_{w,0}^{\text{mem}}}.
\end{align}
By tracing out the memory trees, we obtain the target state.

Step 1 and 3 have circuit depth $O(\log s)$, step 2 have circuit depth $O(\log n)$, and step 4 have circuit depth $O(\log s+\log n)$. Therefore, the total circuit depth of single-word SBM is $O(\log(sn))$.

\begin{algorithm} [H]
\caption{: select($f$) for single-word case}
\label{alg:sbm}
\begin{algorithmic}[1]
\STATE Fanout$\left(T^{\text{idx}}_l\right)$ for all $l\in\{0,1,\cdots,n-1\}$
\STATE Toffoli$\left(r_w,T^{\text{mem}}_w\right)$ for all $w\in\{0,1,\cdots,s-1\}$
\STATE PCNOT$\left(T^{\text{orc}}\right)$
\STATE Toffoli$\left(r_w, T^{\text{mem}}_w\right)$ for all $k\in\{0,1,\cdots,s-1\}$
\STATE Fanout$\left(T^{\text{idx}}_l\right)$ for all $l\in\{1,2,\cdots,n\}$
\end{algorithmic}
\end{algorithm}

\begin{algorithm} [H]
\caption{: Toffoli$\left(r_w,T_w^{\text{mem}}\right)$ }
\label{alg:tof}
\begin{algorithmic}[1]
\STATE  NOT$\left(T^{\text{mem}}_{w,L_n,j}\right)$ for all $r_{w,j}=0$  \quad\quad $\#$ $r_{w,j}$ represents the $(j+1)$th digit of $r_w$

\STATE  CCNOT$\left(T^{\text{mem}}_{w,L_n,2j} , T^{\text{mem}}_{w,L_n,2j+1}; T^{\text{mem}}_{w,L_n-1,j}\right)$ for all $j\in\{0,1,\cdots n/2-1\}$

\STATE  NOT$\left(T^{\text{mem}}_{w,L_n,j}\right)$ for all $r_{w,j}=0$  \quad\quad $\#$ $r_{w,j}$ represents the $(j+1)$th digit of $r_w$

\STATE \textbf{for} $\tilde{m}=1$ to $L_n-1$:
\STATE \quad$m=L_n-\tilde m$

\STATE \quad CCNOT$\left(T^{\text{mem}}_{w,m,2j} , T^{\text{cell}}_{w,m,2j+1}; T^{\text{mem}}_{w,m-1,j}\right)$ for all $j\in\{0,1,\cdots 2^{m-1}-1\}$

\STATE \textbf{end}

\STATE \textbf{for} $\tilde{m}=1$ to $L_n-2$:
\STATE \quad$m=L_n-\tilde m$

\STATE \quad CCNOT$\left(T^{\text{mem}}_{w,m,2j} , T^{\text{cell}}_{w,m,2j+1}; T^{\text{mem}}_{w,m-1,j}\right)$ for all $j\in\{0,1,\cdots 2^{m-1}-1\}$

\STATE \textbf{end}

\STATE  NOT$\left(T^{\text{mem}}_{w,L_n,j}\right)$ for all $r_{w,j}=0$  \quad\quad $\#$ $r_{w,j}$ represents the $(j+1)$th digit of $r_w$

\STATE  CCNOT$\left(T^{\text{mem}}_{w,L_n,2j},T^{\text{mem}}_{w,L_n,2j+1}; T^{\text{mem}}_{w,L_n-1,j}\right)$ for all $j\in\{0,1,\cdots n/2-1\}$

\STATE  NOT$\left(T^{\text{mem}}_{w,L_n,j}\right)$ for all $r_{w,j}=0$

\end{algorithmic}
\end{algorithm}

\subsection{Multi-word cases for PUM and SBM} \label{sec:multi}

In below, we discuss the generalization of PUM and SBM to multi-word cases. The basic idea is to implement fanout (Alg.~\ref{alg:fan}) to obtain multiple ``copies'' of index register, and then query the oracles for different words in parallel, and then implement fanout again for uncomputation.
We take PUM for Lemma~\ref{lm:1} as an example, and same idea can be straightforwardly applied to SBM for Lemma~\ref{lm:2}.

For PUM with $n\geqslant1$, we introduce another $n$ registers, and denote the $l$th register as $I_l$. Each $I_l$ contains $\tilde n$ ancillary qubits with initial state $|0\rangle^{\otimes\tilde n}$. The basis of word register is $|z\rangle\equiv|z_n\rangle\cdots|z_2\rangle|z_1\rangle=\bigotimes_{l=n}^1|z_{l}\rangle$. We can represent the input state of the PUM as
\begin{align}
|\tilde\psi_{\text{input}}^{\text{sbm}}\rangle=\sum_{k=0}^{d-1}\sum_{z=0}^{2^n-1}\psi_{k,z}|k\rangle\bigotimes_{l=n}^{1}|0\rangle^{\otimes \tilde n}_{I_l}|z_{l}\rangle.
\end{align}

The entire query process is illustrate in Fig~\ref{fig:para}.
 In the first step, we fanout the state of index register to each $I_l$, and obtain
\begin{align}
|\tilde\psi_{\text{step 1}}^{\text{sbm}}\rangle=\sum_{k=0}^{d-1}\sum_{z=0}^{2^n-1}\psi_{k,z}|k\rangle\bigotimes_{l=n}^{1}|k\rangle_{I_l}|z_{l}\rangle.
\end{align}
In the second step, we query single-word PUMs select($\hat U_l$) for all $1\leqslant l\leqslant n$ in parallel with $|k\rangle_{I_l}|z_{l}\rangle$ as the input state. The output state is
\begin{align}
|\tilde\psi_{\text{step 2}}^{\text{sbm}}\rangle=&\sum_{k=0}^{d-1}\sum_{z=0}^{2^n-1}\psi_{k,z}|k\rangle\bigotimes_{l=n}^{1}\text{select}(\hat U_l)|k\rangle_{I_l}|z_{l}\rangle\\
=&\sum_{k=0}^{d-1}\sum_{z=0}^{2^n-1}\psi_{k,z}|k\rangle\bigotimes_{l=n}^{1}|k\rangle_{I_l}\hat U_l(k)|z_{l}\rangle.
\end{align}

\begin{figure}[h]
	\centering
	\includegraphics[width=.6\columnwidth]{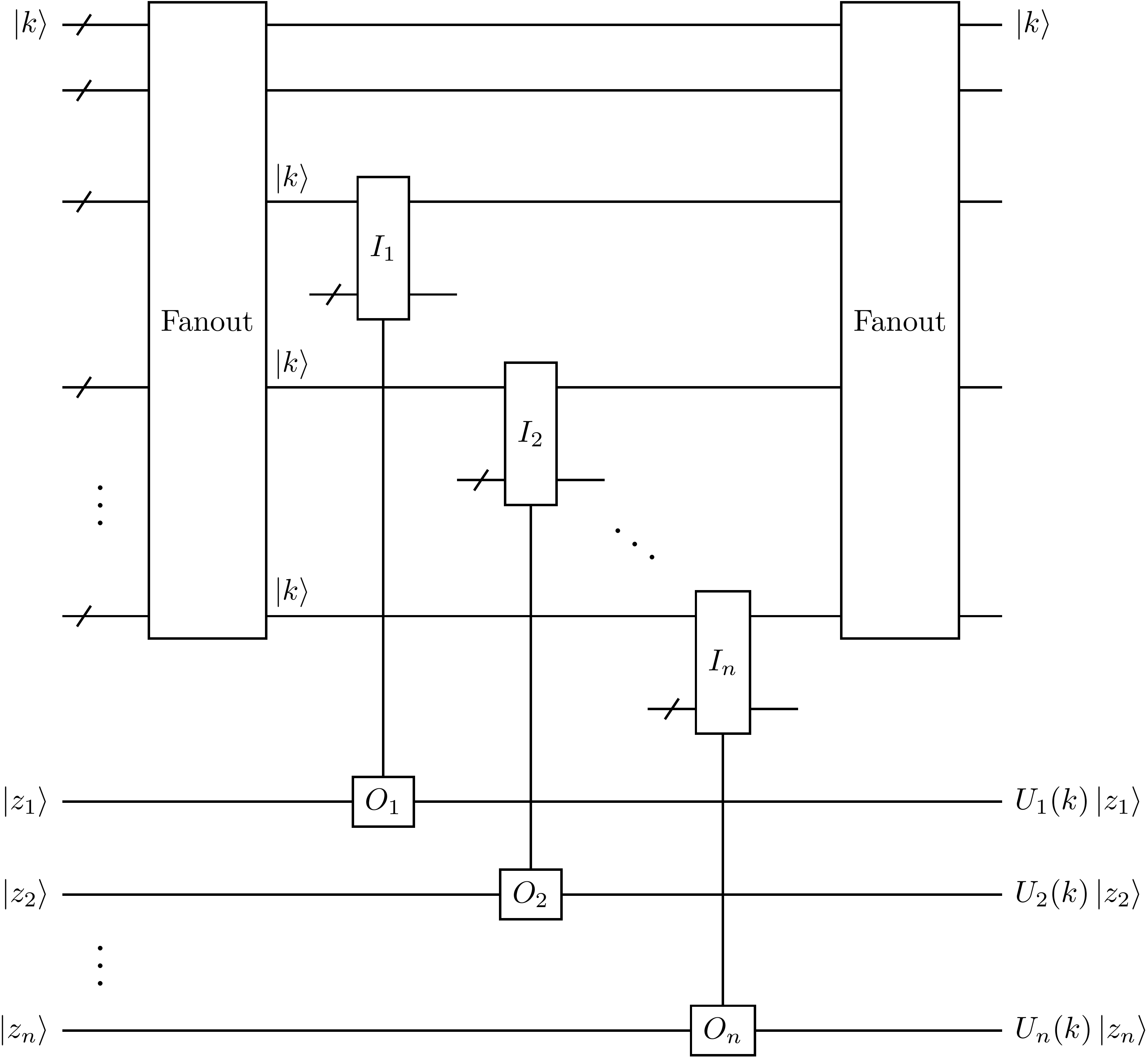}
	\caption{The query of multi-word PUM is realized by querying many single-word PUM in parallel. $I_l$ represents the index register part of the single-word oracle select($\hat U_l$), and $O_l$ represents the word register part of the single-word oracle select($\hat U_l$).}
	\label{fig:para}
\end{figure}

Finally, we repeat the fanout operation in the first step to uncompute $I_l$. The state then becomes
\begin{align}
|\tilde\psi_{\text{step 3}}^{\text{sbm}}\rangle=&\sum_{k=0}^{d-1}\sum_{z=0}^{2^n-1}\psi_{k,z}|k\rangle\bigotimes_{l=n}^{1}|0\rangle^{\otimes \tilde n}_{I_l}\hat U_l(k)|z_{l}\rangle.
\end{align}
Recall that $U(k)=\bigotimes_{l=n}^{1} U_{l}(k)$. By tracing out registers $I_l$, we obtain
\begin{align}
|\tilde\psi_{\text{output}}^{\text{sbm}}\rangle=&\sum_{k=0}^{d-1}\sum_{z=0}^{2^n-1}\psi_{k,z}|k\rangle\bigotimes_{l=n}^{1}\hat U_l(k)|z_{l}\rangle\\
=&\sum_{k=0}^{d-1}\sum_{z=0}^{2^n-1}\psi_{k,z}|k\rangle\hat U(k)|z\rangle,
\end{align}
which is the expected output. The fanout operation has $O(\log n)$ circuit depth, and the parallel query of each select$(\hat U_l)$ has circuit depth $O(\log d)$. So the total circuit depth is $O(\log(nd))$. Because each select$(\hat U_l)$ requires $O(d)$ ancillary qubits and they are queried in parallel, select$(\hat U_l)$ requires totally $O(nd)$ ancillary qubits. So Lemma.~\ref{lm:1} holds true.

We can apply the similar method to generalize SBM for $\tilde{n}>1$. In that case, the total circuit depth is $O(\log(ns\tilde{n}))$ and $O(ns\tilde{n})$ ancillary qubits are required. So Lemma.~\ref{lm:2} holds true.

\section{Hamiltonian simulation based on qubitizaiton}\label{sec:qubit_sup}
According to Ref.~\cite{Low.19}, Hamiltonian simulation can be realized with
block-encoding.

\begin{lemma}[Theorem 1 in Ref.~\cite{Low.19}]\label{lm:hslcu}
Let $\left(\langle G|\otimes \hat{\mathbb{I}}_{L}\right)\hat{U}\left(|G\rangle\otimes \hat{\mathbb{I}}_{L}\right)=\hat{H}/\alpha\in\mathbb{C}^{2^n\otimes2^{n}}$ be Hermitian for some unitary $\hat{U}\in\mathbb{C}^{(P2^n)\times(P2^n)}$ and some state preparation unitary $\hat{G}|0\rangle=|G\rangle\in\mathbb{C}^{P}$ and $\alpha>0$. Then $e^{-i\hat{H}t}$ can be simulated for time $t$, precision $\varepsilon$ in spectral norm, and failure probability $O(\varepsilon)$, using at most $\Theta(\alpha t+\log(1/\varepsilon))$ queries to controlled-$\hat{G}$, controlled-$\hat{U}$, and their inverses, $O\left(\left(\alpha t+\log(1/\varepsilon)\right)\log(P)\right)$ additional two-qubit quantum gates, and $O(n+\log(P))$ extra qubits.
\end{lemma}

The total circuit depth and space complexity are subject to the cost of querying $\hat{G}$ and $\hat{U}$. Theorem.~\ref{th:hs} in the main text follows directly from Lemma.~\ref{lm:hslcu}.

\noindent\textit{Proof of Theorem.~\ref{th:hs}. }
Let  $\hat{G}$ be a unitary satisfying $\hat{G} |0\rangle^{\otimes \lceil\log_2P \rceil}=\sum_{p=0}^{P -1}\alpha_p|p\rangle$. It can be verified that
\begin{align}
\left(\langle 0|^{\otimes \lceil\log_2P \rceil}{\hat G }^\dag\otimes \hat{\mathbb{I}}_{n}\right)\text{select}(\hat{V})\left(\hat{G} |0\rangle^{\otimes \lceil\log_2P \rceil}\otimes \hat{\mathbb{I}}_{n}\right)=\hat{H}/\alpha.
\end{align}
According to Theorem~\ref{th:ab}, implementing $\hat{G}$ requires circuit depth $O(\log P)$ and $O(P)$ ancillary qubits. Because $\hat G$ is a unitary, controlled-$\hat G$ and its inverse have the sample circuit depth and ancillary qubit complexites. According to Lemma.~\ref{lm:1}, implementing select$(\hat V)$ requires circuit depth $O(\log(nP))$ and $O(nP)$ ancillary qubits. Combining these results with Lemma.~\ref{lm:hslcu}, Theorem.~\ref{th:hs} holds true.

\hspace*{\fill}$\Box$\medskip

\section{Parallel Hamiltonian simulation based on quantum walk}

In Ref.~\cite{Zhicheng.21}, the authors developed a parallel algorithm to improve the circuit depth to doubly logarithmic dependence on the precision $\varepsilon$. Their protocol is based on the access of two oracles
\begin{align}
&\mathcal{O}^b_\text{H}|i\rangle|j\rangle|z\rangle=|i\rangle|j\rangle|z\oplus H_{i,j}\rangle\\
&\mathcal{O}_\text{L}|i\rangle|k\rangle=|i\rangle|L(i,k)\rangle
\end{align}
where $i$ and $j$ are the row and column indexes of $\hat H$, $b$ is the precision of $H$, $|z\rangle$ is a $b$-bit basis, $\oplus$ is the bit-wise XOR, and $L(i,k)$ represents the column index of the $k$th nonzero element in row $i$. As an example, we focus on Hamiltonians in the form of Eq.~\eqref{eq:Hpauli} in the main text with $\hat V_l(p)$ restricted to be Pauli operators (although a more general model has been discussed in~\cite{Zhicheng.21}). In this case, $\mathcal{O}^b_\text{H}$ and $\mathcal{O}_\text{L}$ can be realized with circuit depth $O(nb)$ and $O(n)$ respectively.

The goal is to simulate the evolution $e^{-i\hat Ht}$ to precision $\varepsilon$. According to Ref.~\cite{Zhicheng.21}, the simulation can be implemented with $O(\tau\log(\gamma))$ depth of query to $\mathcal{O}_{\text{H}}^b$ and $\mathcal{O}_{\text{L}}$, $O(1)$ depth query to the quantum state preparation unitary for $O(\gamma)$ dimensional states, and $O\big(\tau\log^2\gamma\log^2n\big)$
extra  depth of single- and two-qubit gates. We have assumed $P=O(\text{poly}(n))$, and defined $\tau\equiv Pt$ and $\gamma \equiv\log(\tau/\varepsilon)$. Ref.~\cite{Zhicheng.21} uses a quantum state preparation method with circuit depth $O(\log^3\gamma)$, achieving the total circuit depth
$O\left(\tau\left(\log^2\gamma\log^2n+\log^3\gamma\right)\right)$. 
On the other hand, with our optimal protocol, the circuit depth for state preparation can be improved to $O(\log\gamma)$ without increasing the space complexity. So the total circuit depth is reduced to
$O\left(\tau\log^2\gamma\log^2n\right)$.

\section{Solving linear system}\label{sec:sls}

Ref.~\cite{Childs.17} developed an algorithm for solving linear system based on Hamiltonian simulation and state preparation, whose circuit depth is polylogarithmic in $1/\varepsilon$. The result can be stated as

\begin{lemma}[Theorem. 3 in~\cite{Childs.17}]\label{lm:ls_1}
 Let $|x\rangle$ be a quantum state proportional to $H^{-1}|b\rangle$. A quantum state $|\tilde x\rangle$ satisfying $\||\tilde x\rangle-|x\rangle\|\leqslant \varepsilon$ can be output with $O\left(\kappa\sqrt{\log(\kappa/\varepsilon)}\right)$ uses of Hamiltonian simulation algorithm that approximates $e^{-i\hat Ht}$ for $t=O(\kappa\log(\kappa/\varepsilon))$ with precision $O\left(\varepsilon/\left(\kappa\sqrt{\log(\kappa/\varepsilon)}\right)\right)$, and $O\left(\kappa\sqrt{\log(\kappa/\varepsilon)}\right)$ queries of state preparation oracle $\hat B$ satisfying $\hat B|0\rangle^{\otimes n}=|b\rangle$.
\end{lemma}

We focus on the Hamiltonian in the form of Eq.~\eqref{eq:Hpauli} in the main text. Based on Lemma.~\ref{lm:ls_1}, we have the following result.
\begin{lemma}\label{lm:ls_2}
Let $\hat{H}$ be a Hermitian expressed as
\begin{align}
\hat{H}=\sum_{p=0}^{P-1}\alpha_p\hat{V}(p)
\end{align}
with $\alpha_p> 0$, $\hat V(p)=\bigotimes_{l=1}^{n}\hat V_l(p)$, $\hat V_l(p)\in\text{SU}(2)$, and condition number $\kappa$.
Let $|b\rangle$ be $d$-sparse quantum state. With only single- and two-qubit gates, the quantum state $|x\rangle$ proportional to $\hat H^{-1}|b\rangle$ can be approximately prepared to precision $\varepsilon$ using $\tilde O(\log(nP)\alpha\kappa^2+\log(nd)\kappa)$ circuit depth and $O(n(P+d\log d))$ qubits, where $\tilde O$ neglects the logarithmic dependence on $\kappa,1/\varepsilon$.
\end{lemma}
\begin{proof}
 According to Theorem.~\ref{th:hs}, each unitary $e^{-iHt}$ for $t=O(\kappa\log(\kappa/\varepsilon))$ can be simulated with precision $O\left(\varepsilon/\left(\kappa\sqrt{\log(\kappa/\varepsilon)}\right)\right)$ with $O(\log (nP)(\alpha \kappa\log(\kappa/\varepsilon)+\log((\kappa\sqrt{\log(\kappa/\varepsilon)})/\varepsilon)))$ circuit depth and $O(nP)$ ancillary qubits. So evolution unitary in Lemma.~\ref{lm:ls_1} accounts for totally $\tilde O(\log (nP)\alpha\kappa^2)$ circuit depth, and $O(nP)$ ancillary qubits. According to Theorem.~\ref{th:sp}, for $d$-sparse $|b\rangle$, each query of the state preparation oracle $\hat B$ has circuit depth $O(\log(nd))$ and $O(nd\log d)$ ancillary qubits. So the state preparation in Lemma.~\ref{lm:ls_1} accounts for totally $\tilde O(\log(nd)\kappa)$ circuit depth and $O(nd\log d)$ ancillary qubits. Combining the costs for Hamiltonian simulation and state preparation, Lemma.~\ref{lm:ls_2} can be obtained.
\end{proof}

Theorem.~\ref{th:ls} in the main text is achieved by setting $d=O(\text{poly}(n))$ for Lemma.~\ref{lm:ls_2}. Moreover, by setting $P,\alpha,d=O(1)$, the circuit depth and qubit number reduce to $\tilde O(\log(n)\text{poly}(\kappa))$ and $O(n)$ respectively.

The comparison of our parallel method to other protocols for solving linear systems are provided in Table.~\ref{tab:sls}. Both quantum sequential and quantum parallel methods use the algorithm in~\cite{Childs.17} with the qubitization technique for Hamiltonian simulation~\cite{Low.19}, and the sparse state preparation method introduced in Theorem.~\ref{th:sp}. ``Quantum sequential'' corresponds to the method in~\cite{Berry.15,Childs.18} for select($\hat V$); ``quantum parallel'' corresponds to the method in Lemma.~\ref{lm:1} for select($\hat V$).
\begin{table}[h]\centering
\caption{Space and time complexities for SLS with $O(\text{poly}(n))$-sparse $|b\rangle$. ``$O(1)$ sparse'' corresponds to $P,\alpha,d=O(1)$. We have defined $\kappa_F\equiv\|\hat H\|_F/\|\hat H\|$ with $\|\cdot\|_F$ the Frobenius norm. \label{tab:sls}}
\begin{ruledtabular}
\begin{tabular}{ccccccc}
Algorithm &Time  &Space &Time ($O(1)$ sparse) &Space ($O(1)$ sparse)\\
\hline
Quantum-inspired~\cite{Chia.20,Gilyen.20}&$\tilde O(\text{poly}(n,\kappa,\kappa_F))$&$O(nNP)$&$\tilde O(\text{poly}(n,\kappa,\kappa_F))$&$O(nN)$ \\
Quantum sequential&$\tilde O( nP\alpha\,\text{poly}(\kappa))$& $O(\log P+\text{poly}(n))$&$\tilde O(n\,\text{poly}(\kappa))$&$O(n)$\\
Quantum parallel&$\tilde O(\log(nP)\alpha\,\text{poly}(\kappa))$&$O(P\,\text{poly}(n))$&$\tilde O(\log(n)\text{poly}(\kappa))$&$O(n)$
\end{tabular}
\end{ruledtabular}
\end{table}

\section{Generalized QRAM for continuous amplitudes}\label{sec:pum_sig_q}
Originally, QRAM is designed for storing the classical binary data $\mathcal{D}=\left\{D_k\right\}_{k=0}^{2^n-1}$ with $D_k\in\{0,1\}$. Querying QRAM performs Eq.~\eqref{eq:qram_0} in the main text:
\begin{align}\label{eq:qram_0}
\text{QRAM}(\mathcal{D})\sum_{k=0}^{2^n-1}\psi_k|k\rangle|0\rangle=\sum_{k=0}^{2^n-1}\psi_k|k\rangle|D_k\rangle.
\end{align}
As noted in~\cite{Giovannetti.08_2}, QRAM can not be directly generalized to storing unknown quantum states with $|D_k\rangle\in\mathbb{C}^2$, due to the non-cloning theorem.
However, there are classical data with continuous values of practical interest, so the generalization of QRAM to continuous amplitude is still important. In other words, we want generalize the transformation in Eq.~\eqref{eq:qram_0}, with known, classical description of $|D_k\rangle\in\mathbb{C}^2$. 

With our PUM in Lemma.~\ref{lm:1}, the QRAM for continuous amplitudes can be realized straightforwardly. By setting $n=1$, $|z\rangle=|0\rangle$ and
\begin{align}
\hat U(k)=&|0\rangle\langle D_k|+|1\rangle\langle D^{\bot}_k|
\end{align}
with $\langle D_k|D_k^{\bot}\rangle=0$, Eq.~\eqref{eq:pum_sig} is equivalent to Eq.~\eqref{eq:qram_0}  with $d=N$. According to Lemma.~\ref{lm:1}, the method has $O(n)$ circuit depth and $O(N)$ ancillary qubits. 

In our scheme, each qubit connects to no more than constant number of other qubits and has a tree-like structure (Fig.~\ref{fig:qram}). This connectivity is similar to the conventional QRAM protocols~\cite{Giovannetti.08,Giovannetti.08_2,Hann.19,Hann.21} for binary data. We notice that Refs.~\cite{Yuan.22,Clader.22} independently proposed two schemes about QRAMs for continuous data, but they have assumed all-to-all connectivity. 

\section{Surface code implementation of general quantum state preparation}

\subsection{Clifford+$T$ decomposition}
Before discussing the physical implementation, we introduce how our protocol can be approximated with Clifford+$T$ gates, which is the elementary gates for fault-tolerant surface code computation.  According to Solovay-Kitaev theorem, approximating a single qubit gate with logarithmic circuit depth with respect to accuracy is possible, but how to optimize the total circuit depth of a multi-qubit operation is more complicated. In this section, we show how to perform the general state preparation with $O(n\log(n/\varepsilon))$ circuit depth using only the Clifford+$T$ gates, where $\varepsilon$ is the accuracy of output quantum state measured by vector $2$-norm. We believe that our circuit depth can not be significantly improved further. 

In Alg.~\ref{alg:ab}, except for $S(\theta_{l,j}, H_{l-1,j},H_{l,2j+1})$ (line $3$) and $\text{Ph}(\arg(a_j), H_{n,j})$ (line $6$), all other steps can be ideally implemented by Clifford+$T$ gates. So the fidelity of stage 1 is also the fidelity for the total preparation fidelity. Below, we show how to approximate $S(\theta_{l,j}, H_{l-1,j},H_{l,2j+1})$ and $\text{Ph}(\arg(a_j), H_{n,j})$ with high accuracy using Clifford+$T$ gates, and bound the fidelity of the output state at stage 1.

We begin with the first $n$ steps (line $1$ to $5$). In general, a partial swap gate $S(\theta, a,b)$ defined in Eq.~\eqref{eq:S} can be decomposed as follows.

\begin{center}
\begin{quantikz}[row sep=1cm]
 \lstick{$a$}& \gate[wires=2]{S(\theta,a,b)}&\qw \\
\lstick{$b$}& \qw &\qw
\end{quantikz}=
\begin{quantikz}
\lstick{$a$}&\qw&\ctrl{1}&\qw&\ctrl{1}&\targ{}&\qw\\
\lstick{$b$}&\gate{R_y(\theta)}&\targ{}&\gate{R_y(-\theta)}&\targ{}&\ctrl{-1}&\qw
\end{quantikz}
\end{center}
Note that in our protocol, the input state of $b$ is always $|0\rangle$. Here, $R_y(\theta)=\begin{pmatrix}\cos\theta/2&-\sin\theta/2\\\sin\theta/2&\cos\theta/2\end{pmatrix}$ is the rotation along $y$ axis for angle $\theta$. With elementary gate set $\{\text{Hard},T\}$ where Hard the Hadamard gate, one can approximate arbitrary single qubit gate to accuracy $\varepsilon'$ using circuit depth $O(\log1/\varepsilon')$~\cite{Selinger.12}. More rigorously, there exist $\widetilde R_y(\theta)=\prod_{i=1}^{D}C_i$ with $C_i\in\{\text{Hard},T\}$, such that $D=O(\log1/\varepsilon')$ and $\|\widetilde R_y(\theta)- R_y(\theta)\|=O(\varepsilon')$. Here, $\|\cdot\|$ represents the operator 2-norm for matrices, or vector 2-norm for vectors.
Because $\text{Hard}^\dag=\text{Hard}$ and $T^\dag= \prod_{i=1}^7T$, its inverse $\widetilde R^\dag_y(\theta)$ satisfying
\begin{align}
\widetilde R^\dag_y(\theta)\widetilde R_y(\theta)=\mathbb{I}_1\label{eq:rr}
\end{align}
can also be constructed with $O(\log1/\varepsilon')$ depth by reversing the order of the Hardamard+$T$ gate sequence, and replacing $T$ by $\prod_{i=1}^7T$. We can therefore approximate $S(\theta,a,b)$ with the following circuit
 \begin{center}
\begin{quantikz}[row sep=1cm]
 \lstick{$a$}& \gate[wires=2]{\widetilde S(\theta,a,b)}&\qw \\
\lstick{$b$}& \qw &\qw
\end{quantikz}=\begin{quantikz}
\lstick{$a$}&\qw&\ctrl{1}&\qw&\ctrl{1}&\targ{}&\qw\\
\lstick{$b$}&\gate{\widetilde R_y(\theta)}&\targ{}&\gate{{\widetilde R}^\dag_y(\theta)}&\targ{}&\ctrl{-1}&\qw
\end{quantikz}
\end{center}
which also has circuit depth $O(\log(1/\varepsilon'))$. $\widetilde S(\theta,a,b)$ satisfies
\begin{subequations}\label{eq:s0010}
\begin{align}
(\widetilde S(\theta,a,b)-S(\theta,a,b))|00\rangle&=0\label{eq:s00}\\
(\widetilde S(\theta,a,b)-S(\theta,a,b))|10\rangle&=\alpha|10\rangle+\beta|01\rangle \label{eq:s10}
\end{align}
\end{subequations}
for some $|\alpha|,|\beta|\leqslant \varepsilon'$. In other words, the output state with input $|00\rangle$ is ideal, and the output with input $|10\rangle$ is still in the subspace with single activated qubit. These properties are key to the low depth Clifford+$T$ decomposition. We simply denote the input state of Alg.~\ref{alg:ab} (all qubits at state $|0\rangle$ except for $H_{0,0}$ at state $|1\rangle$) as $|\text{ini}\rangle$. We also denote $W_l$ and $\widetilde W_l$ as the ideal and approximated unitaries at the $l$th step (line $2$ to $4$ with $l$). To facilitate the discussion, we further define
\begin{align}
|\widetilde\Psi_{l}\rangle&=\prod_{l'=l+1}^{n}\widetilde W_{l'}\prod_{l'=1}^{l}W_{l'}|\text{ini}\rangle.
\end{align}
$|\widetilde\Psi_{n}\rangle$ corresponds to the ideal output state and $|\widetilde\Psi_{0}\rangle$ corresponds to the approximated output state with all $S(\theta,a,b)$ replaced by $\widetilde S(\theta,a,b)$. Our goal is to bound $\left\||\widetilde\Psi_{0}\rangle-|\widetilde\Psi_{n}\rangle\right\|$, which corresponds to the error from step $1$ to $n$. According to triangular inequality, we have
\begin{align}
\left\||\widetilde\Psi_{0}\rangle-|\widetilde\Psi_{n}\rangle\right\|\leqslant\sum_{l=0}^{n-1}\left\||\widetilde\Psi_{l}\rangle-|\widetilde\Psi_{l+1}\rangle\right\|.
\end{align}
The next task is to bound $\left\||\widetilde\Psi_{l}\rangle-|\widetilde\Psi_{l+1}\rangle\right\|$. We notice that
\begin{align}
\left\||\widetilde\Psi_{l}\rangle-|\widetilde\Psi_{l+1}\rangle\right\|&=\left\|\widetilde W_{l+1}\prod_{l'=1}^{l}W_{l'}|\text{ini}\rangle-W_{l+1}\prod_{l'=1}^{l}W_{l'}|\text{ini}\rangle\right\|\label{eq:ww1}
\end{align}
where 
\begin{align}
W_{l+1}\prod_{l'=1}^{l}W_{l'}|\text{ini}\rangle=\sum_{k=0}^{2^{l+1}-1}b_{l+1,k}|1\rangle_{H_{0}}\bigotimes_{l'=1}^{l+1}|(k,l')\rangle'_{H_{l'}}.
\end{align} 
According to Eq.~\eqref{eq:s0010}, it can be verified that 
\begin{align}
\widetilde W_{l+1}\prod_{l'=1}^{l}W_{l'}|\text{ini}\rangle=\sum_{k=0}^{2^{l+1}-1}\widetilde b_{l+1,k}|1\rangle_{H_{0}}\bigotimes_{l'=1}^{l+1}|(k,l')\rangle'_{H_{l'}}
\end{align} 
for some $|\widetilde b_{l+1,k}-b_{l+1,k}|\leqslant\varepsilon'|b_{i,k/2}|$ or $|\widetilde b_{l+1,k}-b_{l+1,k}|\leqslant\varepsilon'|b_{i,(k-1)/2}|$, depending of whether $k$ is even or odd. Eq.~\eqref{eq:ww1} can therefore be further bounded by
\begin{align}
\left\||\widetilde\Psi_{l}\rangle-|\widetilde\Psi_{l+1}\rangle\right\|= \sqrt{\sum_{k=0}^{2^{l+1}-1}|\widetilde b_{l+1,k}-b_{l+1,k}|^2}\leqslant\varepsilon'.
\end{align} 
where we have used $\sum_{k=0}^{2^l-1}|b_{l,k}|^2=1$. Therefore, we have
\begin{align}
\left\||\widetilde\Psi_{0}\rangle-|\widetilde\Psi_{n}\rangle\right\|\leqslant n\varepsilon'.\label{eq:psi0n}
\end{align}
To bound the error of Eq.~\eqref{eq:psi0n} to a constant value $\varepsilon/2$, we only need to take $\varepsilon'=\varepsilon/(2n)$. The corresponding circuit depth for each y-rotation is therefore $O(\log(1/\varepsilon'))=O(\log(n/\varepsilon))$, and the total circuit depth for the first $n$ steps is $O(n\log (n/\varepsilon))$.

The next task is to decompose the phase gate Ph$(\text{arg}(a_j),H_{n,j})$ and bound the error at line $6$ of Alg.~\ref{alg:ab}. One may directly decompose the phase gate with gate set $\{\text{Hard}, T\}$.  But the total error of this naive method could be large.  
Instead, for each node $H_{n,j}$, we introduce an extra ancillary qubit $A_{n,j}$ initialized as $|1\rangle$, and implement the following circuit
 \begin{center}
\begin{quantikz}
\lstick{$H_{n,j}$}&\qw&\ctrl{1}&\qw&\ctrl{1}&\qw\\
\lstick{$A_{n,j}$}&\gate{\widetilde {\text{Ph}}(\text{arg}(a_{j})/2)}&\targ{}&\gate{\widetilde{\text{Ph}}^\dag(\text{arg}(a_{j})/2)}&\targ{}&\qw
\end{quantikz}
\end{center}
where $\widetilde{\text{Ph}}(\theta)$ is the $\{\text{Hard},T\}$ decomposition of $\text{Ph}(\theta)=\begin{pmatrix}e^{-i\theta}&\\&e^{i\theta}\end{pmatrix}$ to accuracy $\varepsilon/2$. Similar to previous steps, both $\widetilde{\text{Ph}}(\text{arg}(a_{j})/2)$ and $\widetilde{\text{Ph}}^\dag(\text{arg}(a_{j})/2)$ have circuit depth $O(\log1/\varepsilon)$ and $\widetilde{\text{Ph}}^\dag(\text{arg}(a_{j})/2)\widetilde{\text{Ph}}(\text{arg}(a_{j})/2)=\mathbb{I}_1$ (global phase is neglected).   $A_{n,j}$ are traced out after this step. If $H_{n,j}$ is at state $|0\rangle$, the above circuit output $|0\rangle$ ideally; if $H_{n,j}$ is at state $|1\rangle$, the output error is bounded by $\varepsilon/2$. The input state at this step always has the following form 
\begin{align}
|\widetilde\Psi_{\text{phase in}}\rangle=\sum_{j=0}^{2^n-1}\tilde b_{n,j}|1\rangle_{H_{0}}\bigotimes_{l=1}^{n}|(k,l)\rangle'_{H_{l}}.
\end{align}
For each basis of the above state, there is only one leaf node $H_{n,j}$ at state $|1\rangle$, so there is only two controlled phase gates having contribution to the error. Therefore, the error of this step can achieve $\varepsilon/2$ with circuit depth $O(\log(1/\varepsilon))$.

Combining our approximation technique for $S(\theta_{l,j}, H_{l-1,j},H_{l,2j+1})$ and Ph$(\text{arg}(a_j))$ above, stage $1$ of Alg.~\ref{alg:ab} (line $1$ to $6$) can be approximated to accuracy $\varepsilon$ with $O(n\log (n/\varepsilon)+\log 1/\varepsilon)=O(n\log (n/\varepsilon))$ layer of Clifford+$T$ gates. Because each elementary gate in other stages can be ideally decomposed into Clifford+$T$ gates with constant circuit depth, the total Clifford+$T$ depth for Alg.~\ref{alg:ab} is also $O(n\log (n/\varepsilon))$. 

Compiling the Clifford+$T$ decomposition also requires extra overhead of classical preprocessing. The classical runtime is $O(N\text{polylog}(n/\varepsilon))$ for sequential compiling. By parallel computing with $O(N)$ space complexity, the compiling takes time $O(\text{polylog}(n/\varepsilon))$.

We notice that Ref~\cite{Clader.22} has achieved $T$ depth of $O(n+ \log(1/\varepsilon))$ based on a circuit with all-to-all qubit connectivity, which is better than ours. On the other hand, Ref~\cite{Clader.22} uses FANOUT-CNOT gates. With naive decomposition of the FANOUT-CNOT gate into 2-qubit Clifford gates, the total Clifford+$T$ depth is quadratically higher than the result obtained here. 

\begin{figure}
	\centering
	\includegraphics[width=1\columnwidth]{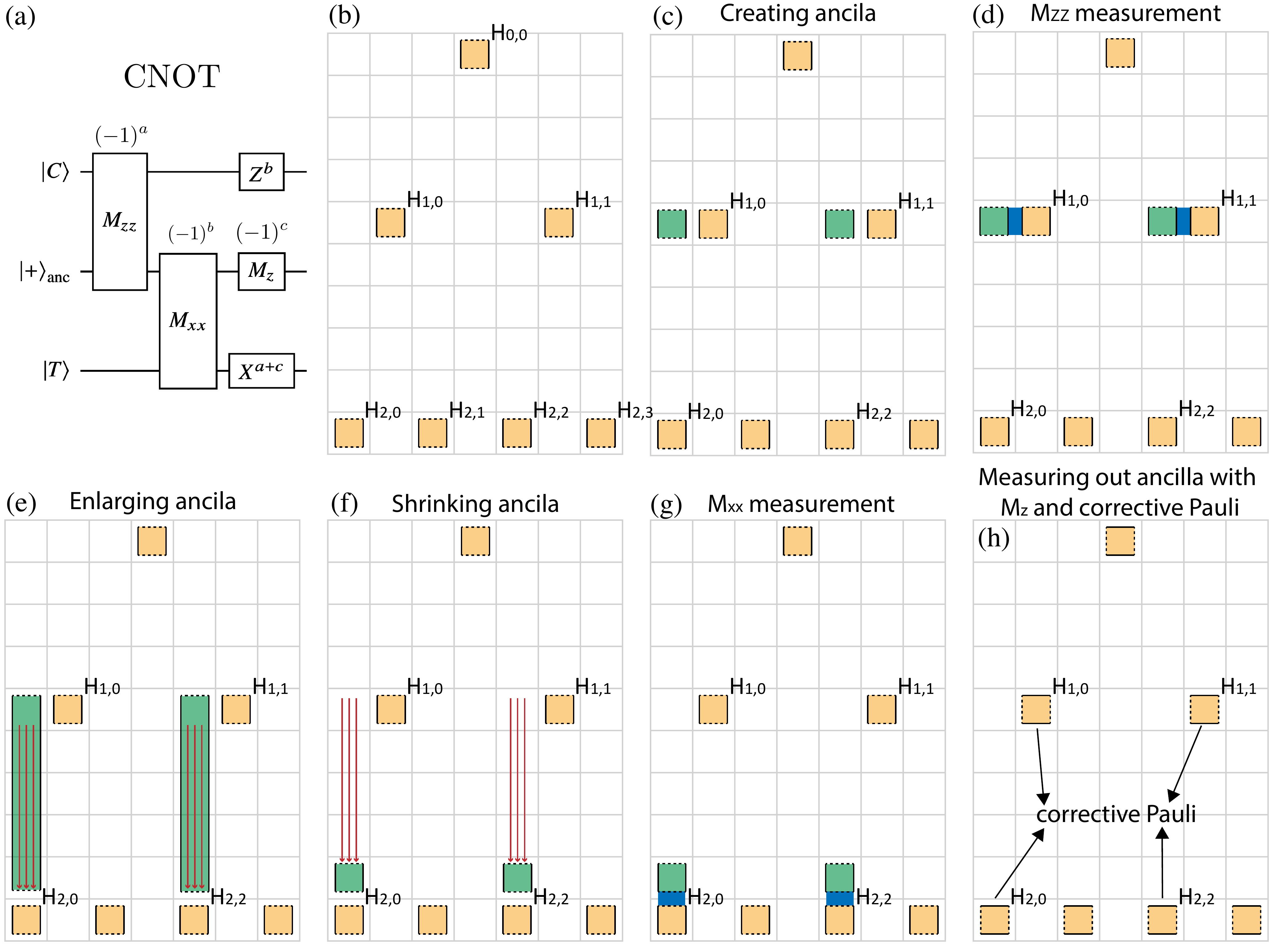}
	\caption{CNOT gate based on lattice surgery, taking  line $2$ of Alg.~\ref{alg:ab} for $l=2$ as an example. Yellow patches represent the logic qubits, and green patches represent the ancillary logici qubits. (a) decomposition of CNOT gate to joint measurements, local measurement and corrective Pauli gates. (b) Layout of the logic qubits at stage 2. (c) Creating ancillary logic qubits. (d) $M_{zz}$ measurement on ancillary qubits and $H_{1,0}$ or $H_{1,1}$. (e) Enlarging ancillary qubits. (f) Shrinking ancillary qubits. (g) $M_{xx}$ measurement on ancillary qubits and $H_{2,0}$ or $H_{2,3}$. (h) Measure out ancillary qubits and corrective Pauli gates on controlled and target qubits of CNOTs.}
	\label{fig:scnot}
\end{figure}

\begin{figure}
	\centering
	\includegraphics[width=1\columnwidth]{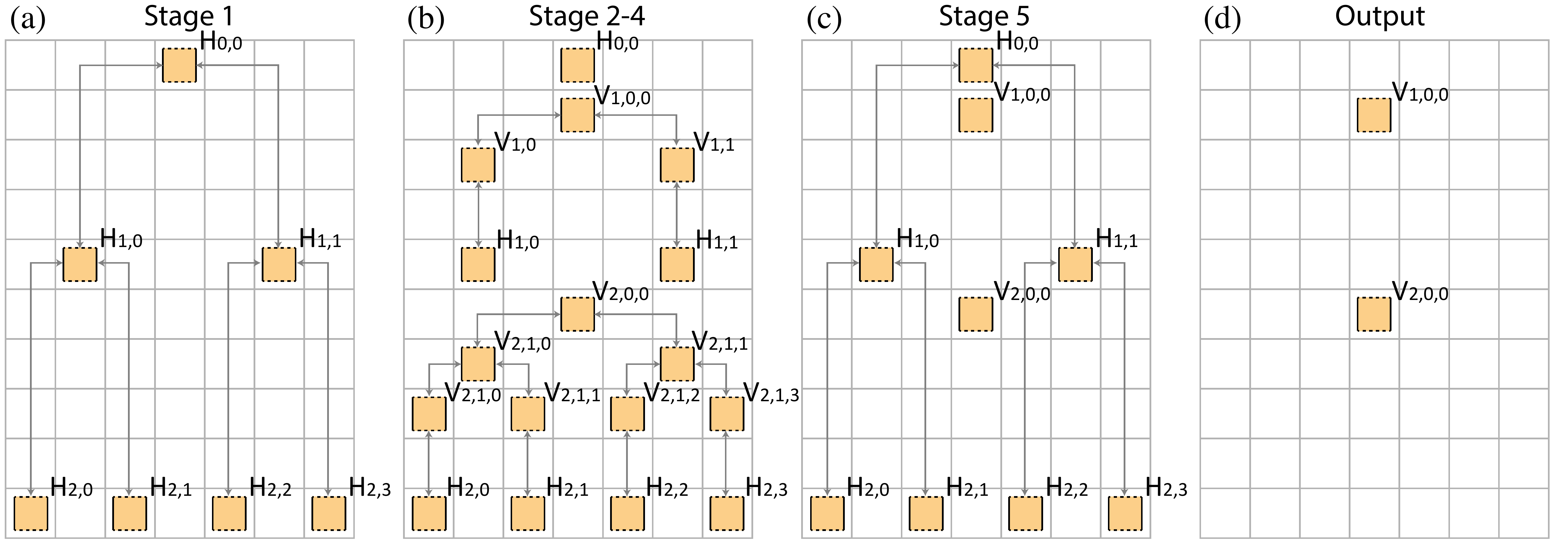}
	\caption{Logic qubit layout for each stage of the general quantum state preparation based on surface code. Yellow patches represent logic qubits, and its smooth and rough boundaries are represented with solid and dashed lines respectively.  An example is given for $n=2$. Grey double-headed arrows represents the paths of ancillary qubits required for CNOT gate implementation. In stage 1, qubits at binary tree $H$ are created. In stage 2, qubits at binary trees $V_{l}$ are created. In stage 5, qubits at binary trees $V_l$, except for root nodes $V_{l,0,0}$, for $1\leqslant l\leqslant n$ are measured and removed. Finally, the target state is encoded in $V_{l,0,0}$ for $1\leqslant l\leqslant n$.}
	\label{fig:surfstg}
\end{figure}

\subsection{Surface code implementation with nearest-neighbor interaction}
Our protocol mainly follows the framework about lattice surgery based surface code computation~\cite{Horsman.12,Litinski.19}. We consider a two-dimensional (physical) qubit array with nearest-neighbor coupling. This qubit array can be realized by, for example, superconducting qubits, trapped ions, and quantum dots. As shown in Fig.~\ref{fig:scnot},~\ref{fig:surfstg}, the array is represented by a board partitioned into many tiles (represented by gray grids). Logic qubits are encoded by surfaces (i.e. a subarray of physical qubits) represented by yellow patches. The ancillary logic qubit for CNOT gate purpose are represented by green patches, which will be introduced later. Note that a surface for each logic qubit can occupy more than one tiles. The solid lines and dashed lines represent the smooth and rough boundaries of the logic qubit respectively. 

Surface code computation are mainly based on $T$ gate and CNOT gate. With distilled magic states, $T$ gates can be realized by performing CNOT gates between the logic qubits and magic states~\cite{Fowler.12}. Below, we focus on how CNOT gates are implemented, and assume that $T$ gates are implemented with standard methods. 

There are in general two methods for CNOT gate implementation, topological braiding~\cite{Fowler.12} and lattice surgery~\cite{Horsman.12}. In this work, we focus on the lattice surgery method, which is based on the circuit shown in Fig.~\ref{fig:scnot}(a). We introduce an ancillary logic qubit initialized at state $|+\rangle$. A joint measurement at $\sigma_z\sigma_z$ basis (denoted as $M_{zz}$) is applied at the controlled and ancillary qubit. Then, another joint measurement at $\sigma_x\sigma_x$ basis (denoted as $M_{zz}$) is applied at the ancillary qubit and target qubit. Next, the ancillary qubit is measured at $Z$ basis (denoted as $M_z$). A Pauli $Z$ gate is applied at the controlled qubit if the measurement outcome of $M_{zz}$ is ``$-1$'', and a Pauli $X$ gate is applied at the target qubit if the measurement outcome of $M_{zz}$ and $M_{z}$ are different.
 
In surface code, all the following operations have constant circuit depth~\cite{Litinski.19}:
\begin{itemize}
\item Creating a qubit initialized at state $|0\rangle$, $|1\rangle$, $|+\rangle$ or $|-\rangle$
\item Single qubit measurement at $\sigma_z$ or $\sigma_x$ basis, and remove the qubit
\item Joint measurement at $\sigma_z\sigma_z$ basis, if two qubits have adjacent smooth boundaries
\item Joint measurement at $\sigma_x\sigma_x$ basis, if two qubits have adjacent rough boundaries
\item Enlarging qubit
\item Shrinking qubit
\end{itemize}
With the elementary operations above, we are ready to show how a CNOT gate in our algorithm is applied with surface code. As an example, we consider line $2$ in Alg.~\ref{alg:ab} with $l=2$, i.e. two parallel CNOT gates with $H_{1,0}$, $H_{1,1}$ as controlled qubits, and $H_{2,0}$, $H_{2,2}$ as target qubits respectively. Each step of the process is shown in Fig.~\ref{fig:scnot} (c)-(h). In the first step [Fig~\ref{fig:scnot} (c)], we create two ancillary qubits initialized as $|+\rangle$ near $H_{1,0}$ and $H_{1,1}$. The smooth boundary of the ancillary qubits and $H_{1,0}$ or $H_{1,0}$ are adjacent to each other. So we are able to apply joint measurements $M_{zz}$ between them [Fig~\ref{fig:scnot} (d)]. 
We then enlarge the ancillary qubits to make their rough boundaries adjacent to the rough boundary of $H_{2,1}$ or $H_{2,3}$ [Fig~\ref{fig:scnot} (e)]. The ancillary qubits are then shrank to the original size [Fig~\ref{fig:scnot} (f)]. In the next step, joint measurements $M_{xx}$ are applied to ancillary qubits and $H_{2,1}$ or $H_{2,3}$ [Fig~\ref{fig:scnot} (g)]. Finally, the ancillary qubits are measured on the $X$ basis and then removed. Corrective Pauli operations are applied on $H_{1,0}$, $H_{1,1}$, $H_{2,0}$ and $H_{2,2}$ accordingly [Fig~\ref{fig:scnot} (h)] according to the rule introduced earlier. 

All operations above has constant circuit depth, so the circuit depth of applying two CNOT gates above is also a constant. It can also be seen from the above example that if one want to apply multiple CNOT gates simultaneously, there should not be crossover between the paths for enlarging and shrinking ancillary qubits. In fact, this can be ensured in all stages of our general state preparation protocol. Fig.~\ref{fig:surfstg} demonstrates all stages of general quantum state preparation. The gray double-headed lines represent the paths of ancillary qubits for CNOT gate implementations. As can be seen, there is no crossover between paths, so there are no extra circuit depth overhead due to the simultaneous implementations of multiple CNOT gates. 

We then illustrate the layout at each stage. In stage $1$, only qubits at binary tree $H$ are created [Fig.~\ref{fig:surfstg}(a)], and the CNOT gates are applied between a qubit and its children. To ensure that there are sufficient space for other binary trees $V_l$, the distances between qubits at horizontal direction are larger when the qubits are closer to the leaf layer. In stage $2$, all qubits at binary tree $V_l$ are created [Fig.~\ref{fig:surfstg}(b)] between the $(l-1)$th layer and the $l$th layer of $H$. Before stage $5$, all qubits at $V_l$ except for the root nodes $V_{l,0,0}$ have been uncomputed, so one can safely measure them out to simplify the layout [Fig.~\ref{fig:surfstg}(e)]. Finally, all qubits at binary tree $H$ are also uncomputed. We measure all qubits at $H$ in $\sigma_z$ basis and remove them. The output state is then encoded in $V_{l,0,0}$.

It can be verified that the height of the entire qubit array is $O(1)+\sum_{i=1}^{n}(O(1)+O(i))=O(n^2)$. The width of the array is $O(N)$. The space complexity is therefore $O(n^2N)$. As we have discussed, parallel CNOT gates will not introduce extra circuit depth overhead to the protocol, even thought the control and target qubits can be spatially far away. So the total circuit depth is $O(n\log(n/\varepsilon))$. 

Finally, we would like to clarify why constant depth of nonlocal CNOT gate is possible with only nearest-neighbor coupling. The key step is enlarging ancillary qubit, that makes the interaction between control and target logic qubits possible. In performing enlarging operations, one should only measure more stabilizers covered by the new patches in the next time step, while each stabilizer can still be measured with constant local operations. The distance between control and target qubits only determines how much stabilizers should be measured (simultaneously), while the circuit depth remains unchanged~\cite{Horsman.12,Litinski.19}.  

\subsection{Surface code implementation with nonlocal interactions}
As has been discussed in the main text, nonlocal entangling gates can be realized in various ways. As an example, Ref.~\cite{Xu.22} has proposed a surface code implementation based on nonlocal CNOT for the purpose of avoiding catastrophe error. In this case, substantial enlarging and shrinking of surfaces can be avoided, and $O(N)$ space complexity is sufficient for implementing Alg.~\ref{alg:ab}.

\section{Approximating general quantum states with sparse states}
As mentioned in the main text, we can speed up the state preparation process by approximating a general quantum state with a sparse quantum state. Let $a^{\text{max}}_j$ be the $j$th largest value of $|a_k|$, and suppose 
\begin{align}
\sum_{j=1}^{d}|a_{j}^{\max}|^2=1-\varepsilon.
\end{align} 
We first define an unnormalized approximation state $|\tilde \psi_{\text{sp}}\rangle=\sum_{k=0}^{2^n-1}c_k|k\rangle$, where 
\begin{eqnarray}
c_{k}= \left\{
\begin{array}{rcl}
a_k     &    &|a_k|\geqslant a_d^{\max} \\
0    &  & |a_k|<a_d^{\max} \\  
\end{array} \right. \ .
\end{eqnarray}
The normalized state of $|\tilde \psi_{\text{sp}}\rangle$ is 
\begin{align}
|\psi_{\text{sp}}\rangle=\frac{1}{\sqrt{\sum_{k=0}^{2^n-1}|c_k|^2}}\sum_{k=0}^{2^n-1}c_k|k\rangle=\frac{1}{\sqrt{\sum_{j=1}^{d}|a^{\max}_j|^2}}\sum_{k=0}^{2^n-1}c_k|k\rangle=\frac{1}{\sqrt{1-\varepsilon}}|\tilde\psi_{\text{sp}}\rangle.
\end{align}
The fidelity between ideal state $|\psi\rangle$ and normalized approximating state $|\psi_{\text{sp}}\rangle$ is 
\begin{align}
F_{\text{sp}}=|\langle\psi|\psi_{\text{sp}}\rangle|^2=\frac{1}{1-\varepsilon}|\langle\psi|\tilde\psi_{\text{sp}}\rangle|^2=\frac{1}{1-\varepsilon}\left(\sum_{k=0}^{2^n-1}|a_kc_k|\right)^2=\frac{1}{1-\varepsilon}\left(\sum_{j=1}^{d}|a^{\max}_j|\right)^2=1-\varepsilon.
\end{align}
According to Theorem.~\ref{th:sp}, $|\psi_{\text{sp}}\rangle$ can be prepared with $O(\log(nd))$ runtime and $O(nd\log d)$ ancillary qubits.


\begin{thebibliography}{57}
\bibitem [{fn()}]{fn}%
  \BibitemOpen
  \href@noop {} {}\bibinfo {note} {In Fig.~\ref{fig:qram}, whd$_k$ and O$_k$
  are not directly connected. But assisted with qubit phd$_k$, Output(word) can
  still be implemented with constant number of single and two qubit
  gates.}\BibitemShut {Stop}%
\end{thebibliography}
\end{document}